\documentclass[journal]{IEEEtran}
\ifCLASSINFOpdf
\else
   \usepackage[dvips]{graphicx}
\fi
\usepackage{url}
\usepackage{graphicx}
\usepackage{stfloats}
\usepackage{fancyhdr}
\usepackage{multirow, cite}
\usepackage{xcolor,colortbl, soul}
\usepackage{tabularx,booktabs}
\usepackage[utf8]{inputenc} 
\usepackage[T1]{fontenc}
\usepackage{url}
\usepackage{ifthen}
\usepackage{cite}
\usepackage[cmex10]{amsmath} 
\usepackage{multicol,lipsum}
\usepackage{amsmath,amssymb,amsfonts,amsthm,bm} 
\usepackage{subfigure}
\usepackage{cite}
\usepackage{amsmath,amssymb,graphicx,epsfig,fancyhdr,amsthm,tabulary,epstopdf}
\usepackage{algorithm}
\usepackage{algorithmic}
\usepackage{pdfsync}
\usepackage{bm}
\usepackage{caption}
\newtheorem{thm}{Theorem}
\newtheorem{cor}{Corollary}
\newtheorem{lem}{Lemma}
\newtheorem{pro}{Proposition}

\theoremstyle{remark}
\newtheorem{rem}{Remark}

\theoremstyle{definition}

\newcommand{\CASE}[1]{\STATE \textbf{case} #1\textbf{:} \begin{ALC@g}}
	\newcommand{\ENDCASE}{\end{ALC@g}}

\newcommand{\DEFAULT}{\STATE \textbf{default:} \begin{ALC@g}}
	\newcommand{\ENDDEFAULT}{\end{ALC@g}}
\newcommand{\DEFAULTLINE}[1]{\STATE \textbf{default:} }

\newcounter{MYtempeqncnt}

\begin{document}

\title{
 $K$-Receiver Wiretap Channel: Optimal Encoding Order and Signaling Design 
}

\author{Yue Qi, \IEEEmembership{Member, IEEE},  Mojtaba Vaezi, \IEEEmembership{Senior Member, IEEE}, \\ and  H. Vincent Poor, \IEEEmembership{Life Fellow, IEEE}  
	\thanks{This paper was presented in part at IEEE International Symposium on Information Theory (ISIT), 2022 \cite{qi2022optimal}.

	Y. Qi and M. Vaezi are with the Department 
	of Electrical and Computer Engineering, Villanova University, Villanova,
	PA 19085 USA (e-mail: yqi@villanova.edu, mvaezi@villanova.edu). 

	H. V. Poor is  with the Department
	of Electrical and Computer Engineering, Princeton University, Princeton,
	NJ 08544 USA (e-mail: poor@princeton.edu).
	His work was supported by the U.S National Science Foundation under Grants CCF-1908308 and CNS-2128448. 
 } 
}
\maketitle 

\begin{abstract}
	The $K$-receiver wiretap channel is a channel model where a {transmitter} broadcasts $K$ independent messages to  $K$  intended receivers while keeping them secret from an eavesdropper. 
	The capacity region of the $K$-receiver multiple-input
	multiple-output (MIMO) wiretap channel has been characterized
	by using dirty-paper coding and stochastic encoding. However, $K$
	factorial encoding orders may need to be enumerated to evaluate
	the capacity region, which makes the problem intractable. In
	addition, even though the capacity region is known, the optimal
	signaling to achieve the capacity region is unknown. In this paper,
	we determine one optimal encoding order to achieve every point
	on the capacity region, and thus reduce the encoding complexity
	$K$ factorial times. We prove that the optimal decoding order
	for the $K$-receiver MIMO wiretap channel is the same as that
	for the MIMO broadcast channel without secrecy. To be specific,
	the descending weight ordering in the weighted sum-rate (WSR)
	maximization problem determines the optimal encoding order.
Next, to reach the border of the secrecy capacity region, we form a
	WSR maximization problem and apply the block successive
	maximization method to solve this nonconvex problem and find
	the input covariance matrices corresponding to each message.
	Numerical results are used to verify the optimality of the encoding
	order and to demonstrate the efficacy of the proposed signaling
	design.
	
\end{abstract} 


\IEEEpeerreviewmaketitle

\section{Introduction}

The \textit{multi-receiver wiretap channel}, depicted in Fig.~\ref{fig:system}, is a channel model in which a transmitter wants to transmit messages of $K$ legitimate receivers while keeping them confidential from an external eavesdropper.  This model, which is an extension of the well-known wiretap channel \cite{wyner1975wire}, is also known variously  as secure broadcasting \cite{ekrem2012degraded},  the wiretap broadcast channel (BC)  \cite{benammar2015secrecy}, and  non-orthogonal multiple access (NOMA) with an eavesdropper \cite{vaezi2019noma}.
In the multiple-input multiple-output (MIMO) multi-receiver wiretap channel, each node can have   an arbitrary number of antennas  \cite{ekrem2011secrecy}. 
The secrecy capacity region of the two-receiver channel was characterized in \cite{liu2010vector, bagherikaram2013secrecy}     
in which secret dirty-paper coding (S-DPC)  is proved to be optimal \cite{bagherikaram2013secrecy}.

\begin{figure}[t]
	\centering
	\includegraphics[width=0.45\textwidth]{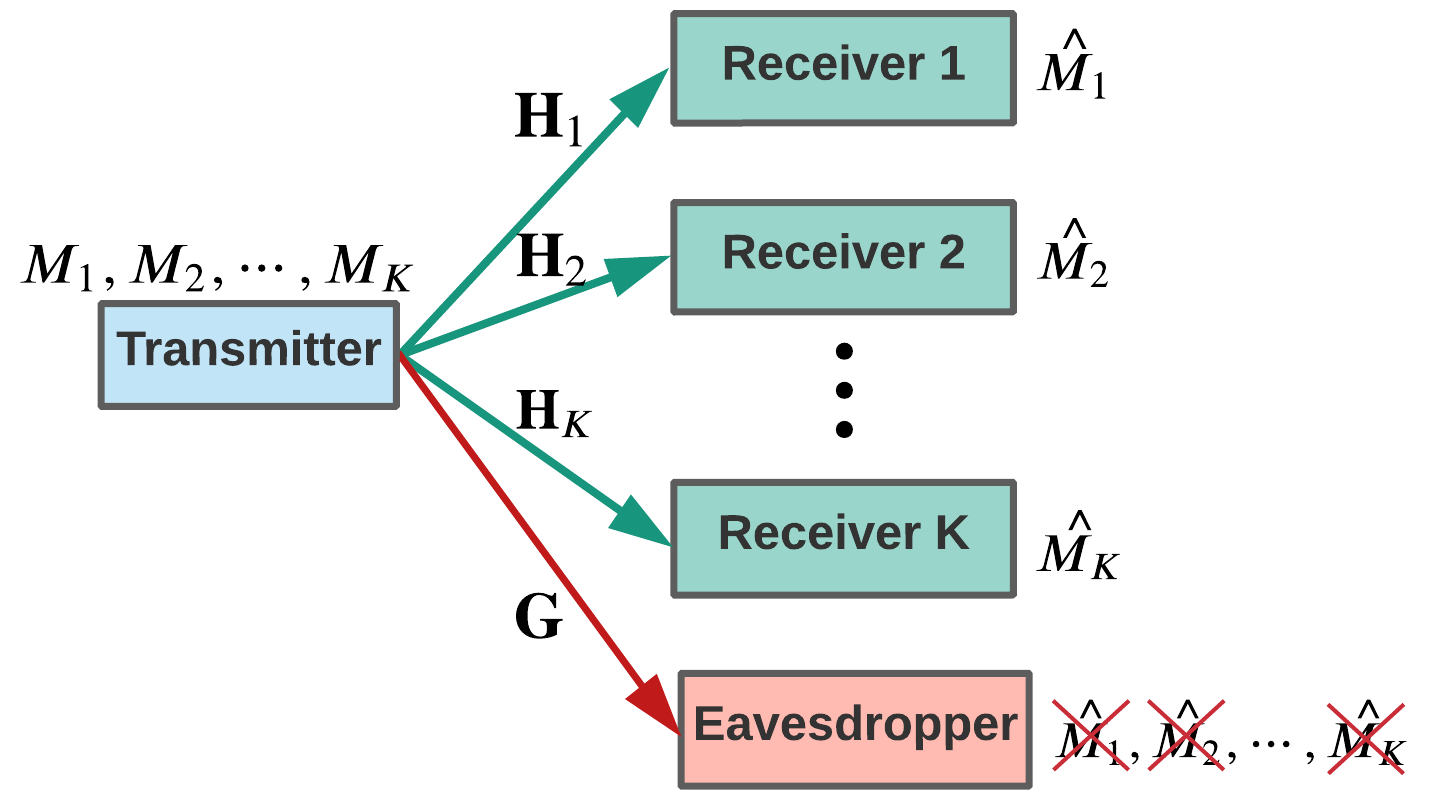}
	\caption{The $K$-receiver wiretap channel.}	
	\label{fig:system}
\end{figure}

Ekrem and Ulukus established the secrecy capacity region of the Gaussian MIMO $K$-receiver wiretap channel in \cite{ekrem2011secrecy}. Interestingly, a different order of encoding for receivers can result in a different achievable region. 
The capacity-achieving  rate region is then characterized by the convex closure of the union of achievable regions obtained by all  one-to-one permutations of encoding order.  Therefore,  to determine
the entire capacity region, all possible one-to-one permutations of receivers, i.e., $K!$  possible encoding orders, may need to be enumerated which makes the
problem intractable especially when the number of receivers becomes large.

In this paper, we prove that such  enumerations can be avoided and we find the optimal encoding order. To this end, we form a weighted sum-rate (WSR) maximization problem for the $K$-receiver wiretap channel,  convert the objective function to another equivalent function, and prove the optimal order for the new optimization problem. 
Our work shows that the optimal encoding order for  secure broadcasting is the same
as that for  insecure broadcasting (i.e., the MIMO-BC without 
secrecy). More specifically, the descending weight ordering in the WSR
maximization problem determines the optimal encoding order \cite{liu2008maximum}. The proof is non-trivial because, unlike the MIMO-BC, the  associated problem in the multi-receiver MIMO wiretap channel is non-convex and  BC to multiple access channel (MAC) duality  \cite{vishwanath2003duality} cannot be directly applied. 
Our proof shows that the problem can be transformed into a set of $K$ parallel wiretap channels with a common eavesdropper. There may also be a relation between these channels and the multi-access wiretap channel \cite{bagherikaram2010secure, xie2013secure}. 

This finding reduces the complexity of evaluating the WSR problem by $K$ factorial. 
When all weights are different, there is only one optimal order of encoding. However,  
when some weights are equal, different encoding orders can give different rate tuples (corner points of the capacity region) while all result in the same WSR. 
Particularly, the encoding order is not important when sum-capacity is the concern, i.e., when all weights are the same. Nonetheless, the order will determine which corner point of the capacity region  will be achieved.

In addition to identifying the encoding order, we propose an algorithm for  designing the transmit covariance matrix of each user to reach the border of the secrecy capacity region of this channel with an arbitrary number of antennas at each node. This is significant because although the secrecy capacity region of the $K$-receiver wiretap channel has been characterized in \cite{ekrem2011secrecy, ekrem2012degraded}, the proof uses  dirty-paper coding (DPC) which is not feasible in practice since the implementation of DPC requires complex random coding and high computational burden \cite{lee2007high}. Other existing solutions either apply only to special cases or fall far from the capacity region. 
 For example, covariance matrix design is known for the secrecy sum-capacity of this channel \cite{park2016secrecy}, but not the whole capacity region.  
	Solutions for the single-antenna receivers are developed in \cite{tang2016low, he2017design}.  Also, zero-forcing (ZF) beamforming has been explored in  \cite{park2016secrecy, tang2016low} for the general case but the result is far away from the capacity  region.   To reach the border of the capacity region, we form a WSR maximization problem and solve it by applying the block successive maximization method (BSMM). The BSMM algorithm  periodically 
	finds optimal solutions for a single block of variables while keeping other blocks of variables fixed at each iteration  \cite{razaviyaynunified, park2015weighted}.  
	This makes the problem convex and tractable at each iteration. 
	
	The main contributions of the work can be summarized as follows:
	
	\begin{itemize}	
		\item We formulate a WSR maximization problem for the  MIMO $K$-receiver wiretap channel.  The problem involves nonconvex functions of  covariance
		matrices of each message {which makes is hard to solve}. Nevertheless, we prove that  the WSR problem has zero duality gap and {thus} the Karush-Kuhn-Tucker (KKT) conditions are necessary for the optimal solutions. 
		This claim helps to prove the optimal encoding order.  
		\item
		We prove that $K!$  enumerations to determine the  capacity region are not necessary, and we find the optimal encoding order. 
		This is completed by transforming the WSR maximization to another equivalent problem, and proving the optimal order for the new optimization problem using the KKT conditions and the dual problem. 
		\item  We next maximize the WSR and find optimal precoding and power allocation matrices (covariance matrices). 
		To tackle this non-convex optimization and  find the optimal covariance matrices for DPC-based rate region, we
		develop an iterative algorithm by applying the BSMM which
		maximizes a lower bound of the problem at each iteration. Using this method, the problem becomes a convex
		problem at each iteration. The proposed algorithm avoids the exhaustive search over all possible covariance matrices.  
		\item  Finally, we study the signaling design for  the communication scenario with two receivers and derive a closed-form solution in this case. 
	\end{itemize}

%

%
The remainder of this paper is organized as follows. In Section~\ref{Sec:II}, we discuss the $K$-receiver wiretap channel model and its special case with two and one receivers.
We introduce the WSR for the $K$-receiver case in Section~\ref{sec:order} and prove the optimal encoding order. We solve the WSR problem by applying the BSMM in Section~\ref{sec:WSR}.  The two-receiver scenario is considered in Section~\ref{sec:tworeceiver}.  
We then present numerical results in  Section~\ref{sec:simulation} and conclude the paper in Section~\ref{sec:conclusion}.

\textit{Notation:} $\rm{tr}(\cdot)$ and $(\cdot)^{\dagger}$ denote trace and Hermitian of matrices. $\mathbb{E}\{\cdot\}$ denotes expectation. ${\rm{diag}}(\lambda_1, 
\dots, \lambda_n)$ represents the diagonal matrix with elements  
$\lambda_1, 
\dots, \lambda_n$. $\mathbf{Q} \succcurlyeq \mathbf{0} $ means that $\mathbf{Q}$ is a positive semidefinite matrix. $[x]^{+}$ gives the maximum value of $x$ and 0. $\mathbf{I}$ is the identity matrix and $|\mathbf{Q}|$ represents the determinant of $\mathbf{Q}$.
Logarithms are taken to the base $e$ and consequently the units of information are nats.


\section{System Model} \label{Sec:II}
We study  a $K$-receiver MIMO wiretap channel with an eavesdropper. The channel model is depicted in 
Fig.~\ref{fig:system}. 
In this problem, a transmitter sends  $K$ independent  messages to $K$ legitimate receivers in the presence of one eavesdropper. Each message $M_k$ for receiver $k$, $k= 1, 2,\dots, K$, should be kept secret from the eavesdropper. There is no cooperation  among legitimate receiver, and the transmitter, receiver $k$, and  the eavesdropper are equipped with $n_t$, $n_k$, and $n_e$ antennas, respectively.   
The received signals at the legitimate receivers and eavesdropper are given by
\begin{subequations}\label{eq:signal model}
	\begin{align} 
		\mathbf{y}_k  &= \mathbf{H}_k  \mathbf{x} +\mathbf{w}_k, \; k= 1, 2,\dots, K\\
		\mathbf{y}_e &= \mathbf{G}\mathbf{x} + \mathbf{w}_e, 
	\end{align}
\end{subequations}
where the channels of the $k$th legitimate receiver $\mathbf{H}_k$  and  eavesdropper $\mathbf{G}$ are ${n_k \times n_t}$ and ${n_e \times n_t}$ complex matrices in which the elements of the channels are drawn from  independent and identically distributed (i.i.d.) complex Gaussian distributions  and  $\mathbf{w}_k$  and  $\mathbf{w}_e$ are i.i.d. complex Gaussian random vectors  whose elements are zero mean and 
unit variance. Perfect knowledge of channel state information at the transmitter is assumed.  
The $K$ independent signals are superimposed and encoded by a Gaussian codebook \cite{ekrem2011secrecy}, and thus, the transmit signal is $\mathbf{x}=\sum_{k=1}^{K}\mathbf{x}_k$ where $\mathbf{x}_k \sim \mathcal{CN}(0, \mathbf{Q}_k)$.  $\mathbf{Q}_k\succcurlyeq  \mathbf{0}$ represents the  covariance matrix of $k$th receiver. 

	The channel input is subject to a power constraint. 
	While it is more common to define capacity regions
	under a total power constraint, i.e., ${\rm tr}(\mathbb{E}\{\mathbf{x}\mathbf{x}^{\dagger}\}) \leq P$, the capacity region of this channel is first proved based on a covariance matrix constraint, i.e., $\mathbb{E}\{\mathbf{x}^{\dagger}\mathbf{x}\} \preccurlyeq \mathbf{S}$  \cite{ekrem2011secrecy, ekrem2012degraded}.
	Once the capacity region is obtained under a covariance constraint $\mathcal{C}(\{\mathbf{H}_k\}_{k=1}^{K}, \mathbf{G}, \mathbf{S})$,  the capacity region
	under more relaxed constraints on the channel inputs  $\mathcal{C}(\{\mathbf{H}_k\}_{k=1}^{K}, \mathbf{G}, P)$ can be obtained by union over compact sets of input covariance matrices as
	\begin{align} \label{constraintP}
	\mathcal{C}(\{\mathbf{H}_k\}_{k=1}^{K}, \mathbf{G}, P) = \bigcup \limits_{\mathbf{S} \succcurlyeq \mathbf{0}, {\rm tr}(\mathbf{S}) \leq P} \mathcal{C}(\{\mathbf{H}_k\}_{k=1}^{K}, \mathbf{G}, \mathbf{S}).
	\end{align}
	%
	In short, the channel input should satisfy
	\begin{align} \label{eq:powercons}
	{\rm tr}(\mathbb{E}\{\mathbf{x}\mathbf{x}^{\dagger}\}) = \sum_{k=1}^{K}{\rm tr}(\mathbf{Q}_k) \leq {\rm tr}(\mathbf{S}) \leq P.
	\end{align}


\subsection{$K$-Receiver  Wiretap Channel} \label{k_receiver}
Let $\pi = [\pi_1, \pi_2, \dots, \pi_K]$ be a permutation function on the set $\{1, \dots, K\}$ and $\pi_k=j$ represent that the $k$th element of that arrangement is $j$. By applying DPC and stochastic
encoding, the achievable secrecy rate at receiver $\pi_k$ under the encoding order $\pi$ is given by \cite{ekrem2011secrecy, park2016secrecy}
\begin{align}\label{eq_KUrate}
R_{\pi_k}=&\log 
\frac{\left|
	\mathbf{I}_{n_k}+\mathbf{H}_{\pi_k}
	\left(\sum_{j=k}^{K}\mathbf{Q}_{\pi_j}\right)\mathbf{H}_{\pi_k}^\dagger
	\right|}{\left|
	\mathbf{I}_{n_k}+\mathbf{H}_{\pi_k}
	\left(\sum_{j=k+1}^{K}\mathbf{Q}_{\pi_j}\right)\mathbf{H}_{\pi_k}^\dagger
	\right|} \notag \\
& \quad \quad - \log 
\frac{\left|
	\mathbf{I}_{n_e}+\mathbf{G}
	\left(\sum_{j=k}^{K}\mathbf{Q}_{\pi_j}\right)\mathbf{G}^\dagger
	\right|}{\left|
	\mathbf{I}_{n_e}+\mathbf{G}
	\left(\sum_{j=k+1}^{K}\mathbf{Q}_{\pi_j}\right)\mathbf{G}^\dagger
	\right|},
\end{align}
in which $k=1, 2,  \dots, K$. Since $\log$ represents the natural logarithm, the units of rates are ${\rm nats/sec/Hz}$ in this paper. The secrecy capacity region under the total power constraint \eqref{eq:powercons} is characterized by the convex closure of the union of the DPC rate region over all possible one-to-one permutations $\pi$ \cite[Theorem 4]{ekrem2011secrecy}. Thus, in general, $K!$ possible  encoding orders may need to be enumerated to determine
the capacity region. 

\subsection{Two-Receiver Wiretap Channel}
This is a special case of the $K$-receiver ($K=2$) scenario in Section~\ref{k_receiver}. The transmitter superimposes two independent messages $M_1$ and $M_2$ intended for receiver 1 and receiver 2, respectively. The messages should be kept confidential from an external {eavesdropper}.
The secrecy capacity region of the two-receiver wiretap channel is given by \cite[Theorem 13]{ekrem2010secure}
\begin{align} \label{eq: level1rate}
\rm{conv} \bigg\{\mathcal{R}_{12} \cup \mathcal{R}_{21} \bigg\}
\end{align}
in which $\rm{conv}$ denotes the convex hull. 
$\mathcal{R}_{12}$   consists of all rate pairs $(R_{1}, R_{2})$ satisfying
\begin{align} \label{scenario1}
&R_{1} \leq   \log|\mathbf{I}_{n_1} + ({\mathbf{I}_{n_1} +\mathbf{H}_1 \mathbf{Q}_2 
	\mathbf{H}_1^\dagger})^{-1}\mathbf{H}_1 
\mathbf{Q}_1 \mathbf{H}_1^\dagger |  \notag\\
& \quad \quad \quad -\log |\mathbf{I}_{n_e} + ({\mathbf{I}_{n_e} +\mathbf{G} \mathbf{Q}_2 \mathbf{G}^\dagger})^{-1}{\mathbf{G} \mathbf{Q}_1 \mathbf{G}^\dagger}|,  \notag \\
&R_{2} \leq  \log|\mathbf{I}_{n_2} + { \mathbf{H}_2 \mathbf{Q}_2 \mathbf{H}_2^\dagger}|-\log|\mathbf{I}_{n_e} +  \mathbf{G} \mathbf{Q}_2 \mathbf{G}^\dagger|.  
\end{align}
To obtain the achievable rate region $\mathcal{R}_{12}$, the encoding order is to encode 
message $M_1$ first and message $M_2$ next, i.e.,  $\pi_1 = 1$, $\pi_2 = 2$. $\mathcal{R}_{21}$ is obtained from $\mathcal{R}_{12}$ by swapping the subscripts 1 and 2 corresponding to different DPC encoding orders, i.e., $\pi_1 = 2$, $\pi_2 = 1$. Further, although the capacity region has been characterized as above, identifying the covariance matrices  to achieve the capacity is still unknown.

\subsection{One-Receiver Wiretap Channel}
It is worth mentioning that when there is only one legitimate receiver ($K=1$), the problem reduces to the MIMO wiretap
channel and \eqref{eq_KUrate} becomes the capacity expression
of the MIMO Gaussian wiretap channel,  established in  \cite{khisti2010secure, oggier2011secrecy, liu2009note}. 
Various linear precoding schemes have been designed for this case, including generalized singular value decomposition \cite{fakoorian2012optimal}, alternating optimization and water filling \cite{li2013transmit}, and rotation modeling \cite{vaezi2017journal, zhang2020rotation}. In this paper, we are interested in the $K$-receiver wiretap channel for $K \geq 2$, but our solutions apply to the case with  $K =1$.

\section{Optimal Encoding Order}\label{sec:order}
The WSR maximization for the $K$-receiver wiretap channel under a total power constraint $P$ is formulated as

\begin{align} \label{eq_WSR}
\varphi(P)=&\max \limits_{\mathbf{Q}_{\pi_k} \succcurlyeq\mathbf{0}}
\sum_{k=1}^K w_{\pi_k} R_{{\pi_k}},  \notag \\
&{\rm s.t.} \quad \sum_{k=1}^K{\rm tr}(\mathbf{Q}_{\pi_k}) \leq P,
\end{align}
in which $R_{k}$ is given in \eqref{eq_KUrate}. The weights $w_{\pi_1}, w_{\pi_2}, \dots, w_{\pi_K}$ are non-negative values adding to one. Solving  problem  \eqref{eq_WSR} for {different permutations $\pi$ results in}  different inner bounds (achievable regions)  of the secrecy capacity region.
The whole capacity region is then determined by enumerating $K!$ possible encoding orders  and finding the union of their convex closures, which makes the
problem intractable if the number of receivers is large.\footnote{Unlike this, in the MIMO-BC without security, there exists an optimal encoding order $\pi$ to maximize the WSR problem \cite{tse1998multiaccess, liu2008maximum}. The proof is determined by the well-known duality between BC and MAC \cite{vishwanath2003duality}.  
}  	

Before finding the optimal order, we first prove that KKT conditions are necessary for the optimal solutions. The Lagrangian of the problem  \eqref{eq_WSR}  is
\begin{align} \label{lagr}
L(\mathbf{Q}_{\pi_1}, \dots, \mathbf{Q}_{\pi_K}, \lambda) =
\sum_{k=1}^K w_{\pi_k} R_{{\pi_k}}-\lambda(\sum_{k=1}^K{\rm tr}(\mathbf{Q}_{\pi_k})-P),
\end{align} 
where $\lambda$ is the Lagrange multiplier related to the total power constraint.  The dual function is a maximization of the Lagrangian 
\begin{align} \label{dual}
g(\lambda)  = \max \limits_{\mathbf{Q}_{\pi_k} \succcurlyeq\mathbf{0}} L(\mathbf{Q}_{\pi_1}, \dots, \mathbf{Q}_{\pi_K}, \lambda),
\end{align}
and the dual problem is given by $\min \limits_{\lambda \geq 0}g(\lambda){ .}$
Usually, if the objective function and the constraints are convex, standard
convex optimization results guarantee that the primal problem \eqref{lagr} and the dual problem  have the same solution. For example, the  {MIMO-BC} (without eavesdropping) \cite{weingartens2006capacity} satisfies Slater’s condition and KKT conditions
are sufficient for optimality \cite{liu2008maximum, qi2021signaling}. When
convexity does not hold, there may exist a duality gap between the primal and dual problems. {That is,   nonconvex optimization problems generally have a nonzero duality gap, and primal and dual problems may not satisfy the KKT conditions. However, with the time-sharing condition, the duality gap can be zero even when the optimization problem is not convex \cite{yu2006dual, park2015weighted} and thus any pair of primal and dual
	optimal points must satisfy the KKT conditions   \cite[Chapter 5, pp. 243]{boyd2004convex}.} We have the following lemma.

\begin{lem} \label{lemma1}
	The problem in \eqref{eq_WSR} has a zero duality gap and the KKT conditions are necessary  for the optimal solution.  
\end{lem}

\begin{proof} See Appendix~\ref{proofLemma}.
	\end{proof}

%

%
%

{In this section, we prove that once the weights are fixed, the encoding order in \eqref{eq_WSR} is determined. 
	Since the proof uses the BC-MAC duality \cite{vishwanath2003duality}, before providing the proof, we first briefly introduce this result  in the following. 
	
	\begin{pro} (BC-MAC Duality \cite{vishwanath2003duality})
	Suppose that receivers $1,  \dots, K$ are encoded sequentially. The  achievable rate of {receiver~$k$ for the MIMO-BC can be computed by}
	\begin{align}
	R^{\rm BC}_{k} = \log 
	\frac{\left|
		\mathbf{I}_{n_k}+\mathbf{H}_{k}
		\left(\sum_{j=k}^{K}\mathbf{Q}_{j}\right)\mathbf{H}_{k}^\dagger
		\right|}{\left|
		\mathbf{I}_{n_k}+\mathbf{H}_{k}
		\left(\sum_{j=k+1}^{K}\mathbf{Q}_{j}\right)\mathbf{H}_{k}^\dagger
		\right|}, \label{eq:BCsum}
	\end{align} 
	and, in {its} dual MIMO-MAC, the achievable  rate is given by
	\begin{align}\label{mac}
	R^{\rm MAC}_{k} = \log 
	\frac{\left|
		\mathbf{I}_{n_t}+	\sum_{j=1}^{k}\mathbf{H}^\dagger_{j}
		\mathbf{\Sigma}_{j}\mathbf{H}_{j}
		\right|}{\left|
		\mathbf{I}_{n_t}+\sum_{j=1}^{k-1}\mathbf{H}^\dagger_{j}
		\mathbf{\Sigma}_{j}\mathbf{H}_{j}
		\right|}.
	\end{align}
	\noindent Given a set of BC covariance matrices $\mathbf{Q}_{k}$, $k=1, \dots, K$, there are  MAC covariance matrices	$\mathbf{\Sigma}_{k}$, $k=1, \dots, K$, such that 	
	$R^{\rm MAC}_{k}=R^{\rm BC}_{k}$, $\forall k$, and $\sum_{k=1}^{K}{\rm tr}(\mathbf{Q}_{k})=\sum_{k=1}^{K}{\rm tr}(\mathbf{\Sigma}_{k})$ via the BC-MAC transformation, and vice versa. {Let us define}
	\begin{align}
	\mathbf{C}_{k}& \triangleq \mathbf{I}_{n_k}+\mathbf{H}_{k}
	(\sum_{j=k+1}^{K}\mathbf{Q}_{j})\mathbf{H}_{k}^\dagger, \notag \\
	\mathbf{D}_{k}&  \triangleq	\mathbf{I}_{n_t}+\sum_{j=1}^{k-1}\mathbf{H}^\dagger_{j}
	\mathbf{\Sigma}_{j}\mathbf{H}_{j}.
	\end{align} 
	Then, the BC and MAC covariance matrices can be expressed as \cite{vishwanath2003duality}
	\begin{align}
	\mathbf{Q}_{k} &= \mathbf{D}^{-1/2}_{k}\mathbf{E}_{k}\mathbf{F}^{\dagger}_{k}\mathbf{C}^{1/2}_{k}	\mathbf{\Sigma}_{k} \mathbf{C}^{1/2}_{k}\mathbf{F}_{k}\mathbf{E}^{\dagger}_{k}\mathbf{D}^{-1/2}_{k},	 \notag \\
	\mathbf{\Sigma}_{k} & =	\mathbf{C}^{-1/2}_{k}\mathbf{E}_{k}\mathbf{F}^{\dagger}_{k}\mathbf{D}^{1/2}_{k}	\mathbf{Q}_{k} \mathbf{D}^{1/2}_{k}\mathbf{F}_{k}\mathbf{E}^{\dagger}_{k}\mathbf{C}^{-1/2}_{k},
	\end{align}	
	in which  $\mathbf{E}_{k}$ and $\mathbf{F}_{k}$ are obtained by singular value decomposition of  $\mathbf{D}_{k}^{-1/2}\mathbf{H}^{\dagger}_ {k}\mathbf{C}_{k}^{-1/2}=\mathbf{E}_{k}\mathbf{\Lambda}_{k}\mathbf{F}_{k}^{\dagger}$, where $\mathbf{\Lambda}_{k}$ is a
	square and diagonal matrix.	
\end{pro}

\begin{thm} \label{lemma2}
	The WSR problem in \eqref{eq_WSR} can be solved by the following equivalent optimization problem:
	\begin{subequations} \label{theorem}
		\begin{align} \label{eq_WSR1}
		&\max \limits_{\mathbf{\Sigma}_{\pi_k} \succcurlyeq\mathbf{0}}
		\sum_{k=1}^K (w_{\pi_k} - w_{\pi_{k-1}}) \times \notag \\
		& \bigg(\log|\mathbf{I}_{n_t} + \sum_{j=k}^{K}\mathbf{H}^{\dagger}_ {\pi_j}\mathbf{\Sigma}_ {\pi_j}\mathbf{H}_ {\pi_j}| - \log|\mathbf{I}_{n_e} + \sum_{j=k}^{K}\mathbf{G}_{\pi_j}^{\dagger}\mathbf{\Sigma}_ {\pi_j}\mathbf{G}_{\pi_j}\bigg), 	 \\ \label{eq:equavalent}
		&{\rm s.t.} \quad \sum_{k=1}^K{\rm tr}(\mathbf{\Sigma}_{\pi_k}) \leq P, \quad k=1,2, \dots, K 
		\end{align}
	\end{subequations}
	{where $w_{\pi_0}\triangleq0$ and $\mathbf{G}_{\pi_j}\triangleq \mathbf{C}_{\pi_j}^{1/2}\mathbf{F}_{\pi_j}\mathbf{E}^{\dagger}_{\pi_j}\mathbf{D}_{\pi_j}^{-1/2}\mathbf{G}^{\dagger}$. Further, the optimal} decoding order $\pi$ is a permutation of the set $\{1,  \dots, K\}$ such that the weights satisfy $w_{\pi_1} \leq w_{\pi_2} \leq  \dots \leq w_{\pi_K}$. The optimal encoding order is the reverse.
\end{thm}

\begin{proof}
	First,  by expanding and rewriting the WSR maximization problem \eqref{eq_WSR}, we have
	\begin{subequations}
		\begin{align}
		&\quad \max \limits_{\mathbf{Q}_{\pi_k} \succcurlyeq\mathbf{0}}	\sum_{k=1}^K w_{\pi_k} R_{\pi_k}  \notag \\
		&=	\max \limits_{\mathbf{Q}_{\pi_k} \succcurlyeq\mathbf{0}}\bigg( \sum_{k=1}^K  w_{\pi_k} \log 
		\frac{\left|
			\mathbf{I}_{n_k}+\mathbf{H}_{\pi_k}
			\left(\sum_{j=k}^{K}\mathbf{Q}_{\pi_j}\right)\mathbf{H}_{\pi_k}^\dagger
			\right|}{\left|
			\mathbf{I}_{n_k}+\mathbf{H}_{\pi_k}
			\left(\sum_{j=k+1}^{K}\mathbf{Q}_{\pi_j}\right)\mathbf{H}_{\pi_k}^\dagger
			\right|}  \notag \\
		&   - \sum_{k=1}^K w_{\pi_k} \log 
		\frac{\left|
			\mathbf{I}_{n_e}+\mathbf{G}
			\left(\sum_{j=k}^{K}\mathbf{Q}_{\pi_j}\right)\mathbf{G}^\dagger
			\right|}{\left|
			\mathbf{I}_{n_e}+\mathbf{G}
			\left(\sum_{j=k+1}^{K}\mathbf{Q}_{\pi_j}\right)\mathbf{G}^\dagger
			\right|}\bigg)  \\
		& \stackrel{(a)}{=} \max \limits_{\mathbf{\Sigma}_{\pi_k} \succcurlyeq\mathbf{0}} \sum_{k=1}^K (w_{\pi_k} - w_{\pi_{k-1}})\log|\mathbf{I}_{n_t} + \sum_{j=k}^{K}\mathbf{H}^{\dagger}_ {\pi_j}\mathbf{\Sigma}_{\pi_j}\mathbf{H}_ {\pi_j}| \notag \\ 
		&   + \max \limits_{\mathbf{Q}_{\pi_k} \succcurlyeq\mathbf{0}} \bigg( - \sum_{k=1}^K w_{\pi_k} \log 
		\frac{\left|
			\mathbf{I}_{n_e}+\mathbf{G}
			\left(\sum_{j=k}^{K}\mathbf{Q}_{\pi_j}\right)\mathbf{G}^\dagger
			\right|}{\left|
			\mathbf{I}_{n_e}+\mathbf{G}
			\left(\sum_{j=k+1}^{K}\mathbf{Q}_{\pi_j}\right)\mathbf{G}^\dagger
			\right|}\bigg)  \label{bc_wsr}\\
		&\stackrel{(b)}{=} \max \limits_{\mathbf{\Sigma}_{\pi_k} \succcurlyeq\mathbf{0}} \sum_{k=1}^K (w_{\pi_k} - w_{\pi_{k-1}}) \log|\mathbf{I}_{n_t} + \sum_{j=k}^{K}\mathbf{H}^{\dagger}_ {\pi_j}\mathbf{\Sigma}_ {\pi_j}\mathbf{H}_ {\pi_j}| \notag \\
		& + \max \limits_{\mathbf{Q}_{\pi_k} \succcurlyeq\mathbf{0}}\bigg( -\sum_{k=1}^K(w_{\pi_k} - w_{\pi_{k-1}}) \log|\mathbf{I}_{n_e} + \mathbf{G}\sum_{j=k}^{K}\mathbf{Q}_{\pi_j}\mathbf{G}^{\dagger}|\bigg)  \label{mac_wsr} \\
		&\stackrel{(c)}{=} 	\max \limits_{\mathbf{\Sigma}_{\pi_k} \succcurlyeq\mathbf{0}} \sum_{k=1}^K (w_{\pi_k} - w_{\pi_{k-1}}) \times  \notag  \\ 
		&\hspace{-2.5mm}  \bigg(\log|\mathbf{I}_{n_t} + \sum_{j=k}^{K}\mathbf{H}^{\dagger}_ {\pi_j}\mathbf{\Sigma}_ {\pi_j}\mathbf{H}_ {\pi_j}| - \log|\mathbf{I}_{n_e} + \sum_{j=k}^{K}\mathbf{G}_{\pi_j}^{\dagger}\mathbf{\Sigma}_ {\pi_j}\mathbf{G}_{\pi_j}|\bigg). \label{lasteq}
		\end{align}
	\end{subequations}
	Here, $(a)$ holds due to \cite[Theorem 1]{liu2008maximum}; $(b)$ holds by expanding the second term in \eqref{bc_wsr} and 
	rearranging the terms to obtain the second term in \eqref{mac_wsr}; and, $(c)$ is obtained by applying the BC-MAC duality and $\mathbf{G}_{\pi_j}\triangleq \mathbf{C}_{\pi_j}^{1/2}\mathbf{F}_{\pi_j}\mathbf{E}^{\dagger}_{\pi_j}\mathbf{D}_{\pi_j}^{-1/2}\mathbf{G}^{\dagger}$. 

		So far, we have proved that \eqref{theorem} is equivalent to \eqref{eq_WSR}. 
		Next, we prove the optimal order of encoding.  To this end, we first write \eqref{lasteq} in an equivalent form and then we use Lagrangian dual to prove the encoding order. Our proof has similarities with the proof of the optimal encoding order for the MIMO-MAC \cite[Theorem 1]{liu2008maximum}. Our problem is however different.  The  MIMO-MAC problem is convex, and thus, KKT conditions are necessary and sufficient for optimality. However, due to the non-convexity of the $K$-receiver wiretap channel in \eqref{eq_WSR}, and its equivalent representation in \eqref{lasteq}, KKT conditions are not sufficient for the optimal solution. Nonetheless, we could still obtain the KKT conditions because, as shown in Lemma~\ref{lemma1},	the problem in \eqref{eq_WSR} has a zero duality gap and the KKT conditions are necessary. 
		
	Rearranging \eqref{lasteq}, we can rewrite the WSR maximization problem as
			\begin{align} \label{macwiretap}
	\max \limits_{\mathbf{\Sigma}_{\pi_k} \succcurlyeq\mathbf{0}}	\sum_{k=1}^K w_{\pi_k} C_{\pi_k}  		
			\end{align}
			in which
			\begin{align} \label{macwiretap1}
			C_{\pi_k} =	&\log 
			\frac{\left|
				\mathbf{I}_{n_t}+\sum_{j=k}^{K}\mathbf{H}^\dagger_{\pi_j}
				\mathbf{\Sigma}_{\pi_j}\mathbf{H}_{\pi_j}
				\right|}{\left|
				\mathbf{I}_{n_t}+	\sum_{j=k+1}^{K}\mathbf{H}^\dagger_{\pi_j}
			\mathbf{\Sigma}_{\pi_j}\mathbf{H}_{\pi_j}
				\right|}    \notag \\
		&	 \quad \quad -  \log 
			\frac{\left|
				\mathbf{I}_{n_t}+\sum_{j=k}^{K}\mathbf{G}^\dagger_{\pi_j}
				\mathbf{\Sigma}_{\pi_j}\mathbf{G}_{\pi_j}
				\right|}{\left|
				\mathbf{I}_{n_t}+\sum_{j=k+1}^{K}\mathbf{G}^\dagger_{\pi_j}
				\mathbf{\Sigma}_{\pi_j}\mathbf{G}_{\pi_j}
				\right|},
		\end{align}
	and $C_{\pi_k} \geq 0$.
	With the above representation, the order of successive decoding is $\pi = [\pi_{1}, \pi_{2}, \dots, \pi_{K}]$ (i.e., $\pi_1$ will be decoded first, followed by $\pi_2$, $\dots$, followed by $\pi_K$). The active rate  constraints corresponding to decoding order $\pi$ are 
	\begin{subequations} \label{prooftheorm1}
	\begin{align}
	&	C_{\pi_K} \leq
	\log{|
		\mathbf{I}_{n_t}+\mathbf{H}^\dagger_{\pi_K}
		\mathbf{\Sigma}^{*}_{\pi_K}\mathbf{H}_{\pi_K}
		|}\notag \\
	& \quad \quad  \quad \quad \quad \quad \quad \quad  - \log{|
		\mathbf{I}_{n_t}+\mathbf{G}^\dagger_{\pi_K}
		\mathbf{\Sigma}^{*}_{\pi_K}\mathbf{G}_{\pi_K}
		|},  \\
		& C_{\pi_{K-1}}  + C_{\pi_K} \leq
	\log{|
		\mathbf{I}_{n_t}+\sum_{j=K-1}^{K} \mathbf{H}^\dagger_{\pi_j}
		\mathbf{\Sigma}^{*}_{\pi_j}\mathbf{H}_{\pi_j}
		|}  \\
	& \quad \quad  \quad \quad \quad \quad \quad \quad  - \log{|
		\mathbf{I}_{n_t}+\sum_{j=K-1}^{K}\mathbf{G}^\dagger_{\pi_j}
		\mathbf{\Sigma}^{*}_{\pi_j}\mathbf{G}_{\pi_j}
		|}, \notag
	\\
	& \quad \quad  \quad \quad \quad \quad \vdots \notag \\
	& C_{\pi_k} +  \dots + C_{\pi_K} \leq
	\log{|
		\mathbf{I}_{n_t}+\sum_{j=k}^{K} \mathbf{H}^\dagger_{\pi_j}
		\mathbf{\Sigma}^{*}_{\pi_j}\mathbf{H}_{\pi_j}
		|} \notag \\
	& \quad \quad  \quad \quad \quad \quad \quad \quad  - \log{|
		\mathbf{I}_{n_t}+\sum_{j=k}^{K}\mathbf{G}^\dagger_{\pi_j}
		\mathbf{\Sigma}^{*}_{\pi_j}\mathbf{G}_{\pi_j}
		|}, \label{bcinequal} 
	\\
	& \quad \quad  \quad \quad \quad \quad \vdots \notag \\
	& C_{\pi_1} +  \dots +C_{\pi_k}+   \dots + C_{\pi_K}  \leq
	\log{|
		\mathbf{I}_{n_t}+\sum_{j=1}^{K} \mathbf{H}^\dagger_{\pi_j}
		\mathbf{\Sigma}^{*}_{\pi_j}\mathbf{H}_{\pi_j}|} \notag \\
	& \quad \quad  \quad \quad \quad \quad \quad \quad  \quad  - \log{|
		\mathbf{I}_{n_t}+\sum_{j=1}^{K}\mathbf{G}^\dagger_{\pi_j}
		\mathbf{\Sigma}^{*}_{\pi_j}\mathbf{G}_{\pi_j}|}, 			\label{sumrate}
	\end{align}
\end{subequations}
where $\mathbf{\Sigma}^{*}_{\pi_j}$, $j = 1, \dots, K$, represent the optimal input covariance matrices that achieve the maximum WSR. 

{Since the duality gap is zero for this problem (we proved this for the equivalent problem in Lemma~\ref{lemma1}), the KKT conditions are necessary. 
}
Thus, let us define the KKT conditions. The primal feasibility of the KKT conditions is defined by
	\begin{align}
	&l_{k} \triangleq \sum_{j=k}^K C_{\pi_j}-\bigg(\log{|
		\mathbf{I}_{n_t}+\sum_{j=k}^{K} \mathbf{H}^\dagger_{\pi_j}
		\mathbf{\Sigma}^{*}_{\pi_j}\mathbf{H}_{\pi_j}
		|} \notag \\
	& \quad \quad  \quad - \log{|
		\mathbf{I}_{n_t}+\sum_{j=k}^{K}\mathbf{G}^\dagger_{\pi_j}
		\mathbf{\Sigma}^{*}_{\pi_j}\mathbf{G}_{\pi_j}
		|}\bigg) \leq 0.
	\end{align}The stationarity KKT condition is given by 
\begin{align}
\frac{\partial \sum_{k=1}^K{w}_k{C}_{\pi_k}}{\partial C_{\pi_k}} - \sum_{k=1}^K\mu_k \frac{\partial l_{k}}{\partial C_{\pi_k}}=0, \; k=1,2, \dots, K,
\end{align}
in which  $\mu_k \geq 0$,  $\forall k$, is the dual feasibility KKT multiplier \cite{boyd2004convex}.
Therefore, we have
\begin{align}
	\begin{bmatrix}
w_{\pi_1} \\
 \vdots\\
 w_{\pi_{K-1}}\\
 w_{\pi_{K}}
\end{bmatrix}=
	\mu_1\begin{bmatrix}
0 \\
\vdots\\
0\\
1
\end{bmatrix} +
	\mu_2\begin{bmatrix}
0 \\
\vdots\\
1\\
1
\end{bmatrix} + \dots + 
	\mu_K\begin{bmatrix}
1 \\
\vdots\\
1\\
1
\end{bmatrix}.
\end{align}
Solving for $\mu_k$, we get $\mu_{K-k+1} = w_{\pi_k} - w_{\pi_{k-1}}$. But, $\mu_k \geq 0$ and thus we must have $w_{\pi_k} - w_{\pi_{k-1}} \geq 0$. Hence, $w_{\pi_1} \leq  w_{\pi_{2}} \leq  \dots \leq w_{\pi_{K}}$, and the decoding order is a permutation $[\pi_{1}, \pi_{2}, \dots, \pi_{K}]$ satisfying
\begin{align}
w_{\pi_1} \leq  w_{\pi_{2}} \leq  \dots \leq w_{\pi_{K}}.
\end{align} 
This completes the proof.

	\end{proof}

To summarize, we have transformed the WSR maximization problem \eqref{eq_WSR} to \eqref{theorem} and then rearranged it into the form of \eqref{macwiretap}. This last representation is equivalent to the original problem, and thus, it satisfies the strong duality due to Lemma~\ref{lemma1}. Hence, the KKT conditions are necessary for \eqref{macwiretap}. Interestingly, the KKT conditions  uniquely produce one of the $K!$ encoding orders, when the weights in the WSR are unique.

	\begin{rem}
		The secrecy capacity region of the general Gaussian MIMO multi-receiver wiretap channel is characterized by the convex closure of the union of the DPC rate region over all possible one-to-one permutations $\pi$ \cite[Theorem 4]{ekrem2011secrecy}. Based on Theorem~\ref{lemma2}, only one encoding order is enough to achieve {each point on the boundary of} the secrecy capacity region, and $K!$  encoding orders are reduced to one order.  The optimal order is the same as that for the MIMO-BC, that is, a receiver	{with a higher} weight in the WSR should be encoded earlier.
	\end{rem}
	

\begin{cor}\label{corollary}
	When all weights are equal, i.e., $w_{\pi_k}= \frac{1}{K}, \forall k\in \{1,\dots,K\}$, the secrecy sum-rate is obtained. In this case, the objective function of the optimization problem in \eqref{eq_WSR1} reduces to $$ \frac{1}{K}\bigg(\log|\mathbf{I}_{n_t} + \sum_{j=1}^{K}\mathbf{H}^{\dagger}_ {\pi_j}\mathbf{\Sigma}_ {\pi_j}\mathbf{H}_ {\pi_j}| - \log|\mathbf{I}_{n_e} + \sum_{j=1}^{K}\mathbf{G}_{\pi_j}^{\dagger}\mathbf{\Sigma}_ {\pi_j}\mathbf{G}_{\pi_j}\bigg).$$
	In this case, any encoding order is optimal  since $w_{\pi_1} \leq  w_{\pi_{2}} \leq  \dots \leq w_{\pi_{K}}$  holds for any order. However, different encoding orders may result in different rate tuples $(R_{\pi_1}, R_{\pi_2}, \dots, R_{\pi_K})$. Yet, all of them will yield the same secrecy sum-rate. 
\end{cor}

\begin{cor}\label{corollary2}
For $K=1$, we get $w_{\pi_1}= 1$ and  the problem reduces to the well-known MIMO wiretap channel  \cite{khisti2010secure, oggier2011secrecy, liu2009note}. In this case, the objective function of the optimization problem in \eqref{eq_WSR1} becomes $$ \log|\mathbf{I}_{n_t} + \mathbf{H}^{\dagger}_ {\pi_1}\mathbf{\Sigma}_ {\pi_1}\mathbf{H}_ {\pi_1}| - \log|\mathbf{I}_{n_e} + \mathbf{G}_{\pi_1}^{\dagger}\mathbf{\Sigma}_ {\pi_1}\mathbf{G}_{\pi_1}|.$$ 
\end{cor}

\section{Signaling Design}\label{sec:WSR}

To reach the border of the secrecy capacity region, we need  to solve the WSR problem in \eqref{eq_WSR}. Unfortunately, this problem is non-convex and challenging to solve, even after finding the optimal encoding order. In this section, we propose a numerical algorithm to solve \eqref{eq_WSR}, i.e., to design transmit covariance matrices.

To this end, we apply BSMM. This algorithm  can be traced back to the alternative inexact block coordinate descent method proposed in \cite{razaviyaynunified}, which cyclically finds the optimal solution for a single block of variables while keeping other block variables fixed at each iteration. Later, a WSR problem of the two-user MIMO-BC with confidential messages was studied in  \cite{park2015weighted}, where BSMM was developed by specifying first-order Taylor expansion for the lower bound of  each convex function. However, the previous approaches \cite{razaviyaynunified, park2015weighted} are not directly applicable to the $K$-receiver wiretap channel as the system models are different, and the number of users is limited \cite{park2015weighted}. Thus, we derive and design the covariance matrices to reach the border of the secrecy capacity region.  For simplicity of derivations and notation, without loss of generality, we assume the encoding order is ascending order, that is $\pi_k=k$.

The BSMM updates covariance matrices by successively optimizing a lower bound of the local approximation of
$f(\mathbf{Q}_1, \mathbf{Q}_2, \dots, \mathbf{Q}_K) = L(\mathbf{Q}_1, \mathbf{Q}_2, \dots, \mathbf{Q}_K, \lambda)$ {given in \eqref{lagr}}. At  iteration $i$ of the algorithm, the variables $\mathbf{Q}^{(i)}_k$, $k = 1, 2,   \dots, K$, are updated by solving the following problem: \cite{razaviyaynunified, park2015weighted}
\begin{align} \label{bsm1}
	\mathbf{Q}^{(i)}_k = {\rm arg} \max_{\mathbf{Q}_k \succcurlyeq 0} f(\mathbf{Q}^{(i)}_1,  \dots, \mathbf{Q}^{(i)}_{k-1},	\mathbf{Q}_k, \mathbf{Q}^{(i-1)}_{k+1},  \dots,  \mathbf{Q}^{(i-1)}_K).
\end{align}

The function $f(\mathbf{Q}_1, \mathbf{Q}_2, \dots, \mathbf{Q}_K)$ can be written as the summation of  convex and concave functions
\begin{align} \label{conv_conca}
f(\mathbf{Q}_1, \mathbf{Q}_2, \dots, \mathbf{Q}_K) =f^{\rm ccv}_{k}(\mathbf{Q}_k, \mathbf{Q}_{\bar{k}}) + f^{\rm cvx}_{k}(\mathbf{Q}_k, \mathbf{Q}_{\bar{k}}),
\end{align}
\noindent where $\mathbf{Q}_{\bar{k}} = (\mathbf{Q}_1,   \dots, \mathbf{Q}_{k-1}, \mathbf{Q}_{k+1},   \dots, \mathbf{Q}_{K})$ represents all covariance matrices excluding $\mathbf{Q}_k$, $f^{\rm ccv}_{k}(\mathbf{Q}_k, \mathbf{Q}_{\bar{k}})$ is a concave function of $\mathbf{Q}_k$ with $\mathbf{Q}_{\bar{k}}$ fixed, and $f^{\rm cvx}_{k}(\mathbf{Q}_k, \mathbf{Q}_{\bar{k}})$ is convex function of $\mathbf{Q}_k$ with $\mathbf{Q}_{\bar{k}}$ fixed.
We can write these two functions as \eqref{krcvcx1} shown at the top of the next page. 

\begin{figure*}[!t] 
	\normalsize
	\setcounter{MYtempeqncnt}{\value{equation}}
	\setcounter{equation}{23}
	\begin{subequations}\label{iterationeq}
		\begin{align}	
		\hline
		&f_k^{\rm ccv}(\mathbf{Q}_k,\mathbf{Q}_{\bar{k}})=w_k 
		\log{|
			\mathbf{I}+({
				\mathbf{I}+\mathbf{H}_k
				\sum_{j=k+1}^{K}\mathbf{Q}_j\mathbf{H}_k^\dagger
			})^{-1}\mathbf{H}_k\mathbf{Q}_k\mathbf{H}_k^\dagger
			|}+\sum_{j=1}^{k-1}w_j\log\big|
		\mathbf{I}+\mathbf{G}
		\sum_{i=j+1}^{K}\mathbf{Q}_i\mathbf{G}^\dagger
		\big|
		-\lambda{\rm tr}(\mathbf{Q}_k),\\
		&f_k^{\rm cvx}(\mathbf{Q}_k,\mathbf{Q}_{\bar{k}})=
		-w_k\log{|
			\mathbf{I}+({
				\mathbf{I}+\mathbf{G}
				\sum_{j=k+1}^{K}\mathbf{Q}_j\mathbf{G}^\dagger
			})^{-1}\mathbf{G}\mathbf{Q}_k\mathbf{G}^\dagger
			|}
		+\sum_{j=1}^{k-1}w_j\log\frac{\left|
			\mathbf{I}+\mathbf{H}_j
			\sum_{i=j}^{K}\mathbf{Q}_i\mathbf{H}_j^\dagger
			\right|}{\left|
			\mathbf{I}+\mathbf{H}_j
			\sum_{i=j+1}^{K}\mathbf{Q}_i\mathbf{H}_j^\dagger
			\right|}\notag \\ 
		&\quad\quad\quad\quad\quad-\sum_{j=1}^{k-1}w_j\log\big|
		\mathbf{I}+\mathbf{G}
		\sum_{i=j}^{K}\mathbf{Q}_i\mathbf{G}^\dagger
		\big|
		+\sum_{j=k+1}^{K}w_jR_j-\lambda{\rm tr}(\sum_{j=1}^{K}\mathbf{Q}_j-\mathbf{Q}_k- P).
		%
		%
		%
		\end{align} \label{krcvcx1} 
	\end{subequations}
	\hrulefill
	\vspace*{4pt}
\end{figure*}
%
%
%
%

For the $i$th iteration, the convex function $f^{\rm cvx}_{k}(\mathbf{Q}_k, \mathbf{Q}^{(i)}_{\bar{k}})$ can be lower-bounded by its gradient \cite{razaviyaynunified}
\begin{align} \label{Talyor}
	f^{\rm cvx}_{k}(\mathbf{Q}_k, \mathbf{Q}^{(i)}_{\bar{k}}) \geq& 	f^{\rm cvx}_{k}(\mathbf{Q}^{(i)}_k, \mathbf{Q}^{(i)}_{\bar{k}}) + {\rm {tr}}[\mathbf{A}_{k}^{(i)\dag}	(\mathbf{Q}_k - \mathbf{Q}^{(i)}_k)],
\end{align}
in which $\mathbf{Q}^{(i)}_{\bar{k}} =( \mathbf{Q}^{(i)}_1,  \dots, \mathbf{Q}^{(i)}_{k-1},	\mathbf{Q}^{(i-1)}_{k+1},  \dots, \mathbf{Q}^{(i-1)}_K)$, where the first $k-1$ covariance matrices have been optimized in the $i$th iteration, while the $k+1$ to $K$ covariance matrices are from the previous $(i-1)$th iteration waiting to be optimized. The power price matrix is a partial derivative with respect to $\mathbf{Q}_k$  which is given by
\begin{align}
	\mathbf{A}_{k}^{(i)} =& \triangledown_{\mathbf{Q}_k} f^{\rm cvx}_{k}(\mathbf{Q}_k, \mathbf{Q}_{\bar{k}})|_{\mathbf{Q}^{(i)}_{k}, \mathbf{Q}^{(i)}_{\bar{k}}}.   \label{eq:A1}
\end{align}
Following the matrix identity \cite{magnus2019matrix},
$$\frac{\partial \log |\mathbf{A+BXC}|}{\partial \mathbf{X}}  = [(\mathbf{C(A+BXC)}^{-1})\mathbf{B}]^\dagger,$$  the partial derivatives of	$f^{\rm cvx}_{k}(\mathbf{Q}_k)$ with respect to $\mathbf{Q}_k$ are
\begin{align} \label{derivative}
	\mathbf{A}_k =& \triangledown_{\mathbf{Q}_k} f^{\rm cvx}_{k}(\mathbf{Q}_k, \mathbf{Q}_{\bar{k}})     \notag \\
	&=	-w_k \mathbf{G}^{\dagger}{(
		\mathbf{I}+	\mathbf{G}
		\sum_{j=k}^{K}\mathbf{Q}_j\mathbf{G}^\dagger)^{-1}}\mathbf{G} \notag \\
	&+\sum_{j=1}^{k-1}w_j\mathbf{H}_j^\dagger{(\mathbf{I}+\mathbf{H}_j
		\sum_{i=j}^{K}\mathbf{Q}_i\mathbf{H}_j^\dagger)^{-1}}\mathbf{H}_j \notag\\ 
		&-\sum_{j=1}^{k-1}w_j\mathbf{H}_j^\dagger{(\mathbf{I}+\mathbf{H}_j
			\sum_{i=j+1}^{K}\mathbf{Q}_i\mathbf{H}_j^\dagger)^{-1}}\mathbf{H}_j \notag\\ 
	&-\sum_{j=1}^{k-1} w_j
	\mathbf{G}^\dagger(\mathbf{I}+\mathbf{G}
	\sum_{i=j}^{K}\mathbf{Q}_i\mathbf{G}^\dagger)^{-1}\mathbf{G}.	
\end{align} 
A lower bound on $f(\mathbf{Q}_1, \mathbf{Q}_2, \dots, \mathbf{Q}_K)$ for $i$th iteration  can be obtained by substituting the right-hand terms in \eqref{Talyor} and \eqref{conv_conca} as follows:
\begin{align} \label{reformuP21}
f(\mathbf{Q}^{(i)}_1, \dots, \mathbf{Q}^{(i)}_{k-1},	\mathbf{Q}_k, \mathbf{Q}^{(i-1)}_{k+1},  \mathbf{Q}^{(i-1)}_K)  &\geq	 \notag \\
f^{\rm ccv}_{k}(\mathbf{Q}_k, \mathbf{Q}^{(i)}_{\bar{k}})+
f^{\rm cvx}_{k}(\mathbf{Q}^{(i)}_k, \mathbf{Q}^{(i)}_{\bar{k}})&
 \notag \\
+ {\rm {tr}}[\mathbf{A}_{k}^{(i)\dag}	(\mathbf{Q}_k - \mathbf{Q}^{(i)}_k)].&
\end{align}
On substituting \eqref{iterationeq} into the right-hand side of the inequality in \eqref{reformuP21} and omitting the constant terms $f^{\rm cvx}_{k}(\mathbf{Q}^{(i)}_k, \mathbf{Q}^{(i)}_{\bar{k}})$ and ${\rm {tr}}(\mathbf{A}_{k}^{(i)\dag}	\mathbf{Q}^{(i)}_k)$, the optimization of the remaining terms in each iteration can be expressed as  \eqref{reformuP3}, which is a convex problem.


\begin{figure*}[!tb]
	\hrulefill
	\normalsize
	\setcounter{MYtempeqncnt}{\value{equation}}
	\setcounter{equation}{28}
	\begin{align}  \label{reformuP3}
	\mathbf{Q}^{(i)}_k= {\rm arg}\max \limits_{\mathbf{Q}_k} \; w_k 
	\log{|
		\mathbf{I}+({
			\mathbf{I}+\mathbf{H}_k
			\sum_{j=k+1}^{K}\mathbf{Q}_j\mathbf{H}_k^\dagger
		})^{-1}\mathbf{H}_k\mathbf{Q}_k\mathbf{H}_k^\dagger
		|}+\sum_{j=1}^{k-1}w_j\log\big|
	\mathbf{I}+\mathbf{G}
	\sum_{i=j+1}^{K}\mathbf{Q}_i\mathbf{G}^\dagger
	\big|
	-{\rm tr}[(\lambda\mathbf{I}-\mathbf{A}_{k}^{(i)\dag})\mathbf{Q}_k], 
	\end{align} 
	\hrulefill
\vspace*{4pt}
\end{figure*}

\begin{rem}[\textit{Nash equilibrium \cite{scutari2013decomposition, park2015weighted}}]
	Any fixed point after the $i$th iteration $\mathbf{Q}^{(i)}$ is a Nash equilibrium of the optimization with power price game where each receiver tries to 	maximize the summation between its own benefit $f^{\rm ccv}_{k}(\mathbf{Q}_k, \mathbf{Q}^{(i)}_{\bar{k}})$  and the power price ${\rm Re}({\rm  {tr}}(\mathbf{A}_{k}^{(i)\dag}	\mathbf{Q}_k))$ along the direction of the gradient 
	\begin{align} \label{bsmm}
	{\rm arg} \max_{\mathbf{Q}_k \succcurlyeq 0} 
	f^{\rm ccv}_{k}(\mathbf{Q}_k, \mathbf{Q}^{(i)}_{\bar{k}})	+{\rm  {tr}}(\mathbf{A}_{k}^{(i)\dag}	\mathbf{Q}_k),
	\end{align}
\end{rem}
\noindent in which $\mathbf{A}_k = \triangledown_{\mathbf{Q}_k} f^{\rm cvx}_{k}(\mathbf{Q}_k, \mathbf{Q}_{\bar{k}})|_{\mathbf{Q}_k = \mathbf{Q}^{(i)}}$. Thus, $\mathbf{Q}^{(i)}$ is unilaterally optimal for the problem \eqref{bsmm} \cite[Proposition~2]{scutari2013decomposition, park2015weighted}.

The optimal numerical solution for  \eqref{reformuP3} can be solved by using any standard convex tool. 
  We choose the rotation-based linear precoding \cite{vaezi2017journal, zhang2020rotation} to reach the optimality. 

Algorithm~\ref{alg:WSRalgorithm} summarizes the WSR maximization, where
$\epsilon_1$ and $\epsilon_2$ are the bisection search accuracy and convergence tolerance, respectively. When $w_k=1$, the other weights are zero and the problem reduces to the  wiretap channel where receiver $k$ is the legitimate receiver. According to Lemma~\ref{lemma1}, Algorithm~\ref{alg:WSRalgorithm} converges to a KKT point and satisfies the necessary conditions for the optimal solution. However, the initial point may affect the speed of convergence.  $\mathbf{Q}_k^{(0)}:=\frac{P}{Kn_t}\mathbf{I}$ or $\mathbf{Q}_k^{(0)}:=\mathbf{0}$ are two possible choices.
 
 \begin{algorithm}[t]
	\caption{WSR maximization
	}\label{alg:WSRalgorithm}
	\begin{algorithmic}[1]
		\STATE	inputs:  $\lambda^{\max}:=10$, $\lambda^{\min}:=0.01$,	$\epsilon_1$, 		$\epsilon_2$, $(w_1, \dots, w_K)$
		\WHILE {$\lambda^{\max}-\lambda^{\min} > \epsilon_1$}
		\STATE$\lambda:=(\lambda^{\max}+\lambda^{\min})/2$; $i:=0$;
		\STATE
		$\mathbf{Q}_1^{(0)}=\mathbf{Q}_2^{(0)}= \dots=\mathbf{Q}_k^{(0)}:=\frac{P}{Kn_t}\mathbf{I}$; $R^{(0)}:=0$;
		\WHILE 1
		\STATE $i=i+1$;
		\FOR {$k=1:K$}
		\STATE Solve for $\mathbf{Q}_k^{(i)}$  in 
		\eqref{reformuP3};
		\STATE Compute $R_{k}$ using \eqref{eq_KUrate};
		\ENDFOR
		\STATE $R^{(i)}:=\sum_{k=1}^K w_k R_{k}$; 
		\IF {${{\rm abs}(R^{(i)} - R^{(i-1)})} < \epsilon_2$}
		\STATE break;
		\ENDIF
		\IF {$\sum_{k=1}^K{\rm tr}(  \mathbf{Q}_k^{(i)}) < P$}
		\STATE $\lambda^{\max}:=\lambda$;
		\ELSE
		\STATE $\lambda^{\min}:=\lambda$;
		\ENDIF
		\ENDWHILE
		\ENDWHILE
		\STATE outputs: $\lambda^{*}:=\lambda$, $R_{k}^{*}:=R_{k}$, and
		$\mathbf{Q}_k^{*}=\mathbf{Q}_k^{(i)}$.
	\end{algorithmic}
\end{algorithm}

\textit{Complexity analysis.}  The computation of matrix multiplications and the matrix inverse is of complexity $\mathcal{O}(m^3)$ where {$m= \max(n_t, n_k, n_e)$.} The number of iterations of the BSMM is $\mathcal{O}(1/\epsilon_1)$, and the bisection search requires  $\mathcal{O}(\log(1/\epsilon_2))$. In general, one WSR point achieved by Algorithm~\ref{alg:WSRalgorithm} has  complexity  $\mathcal{O}(\frac{m^3}{\epsilon_3 }\log({1}/{ \epsilon_2}))$ for each group of weights \cite{park2015weighted, qi2020secure}. Each group can realize one point on the boundary. To find the capacity, another exhaustive search over weights is needed.  	The capacity region achieved by DPC in \cite{ekrem2011secrecy} is found by using an exhaustive search over all permutations  where the covariance matrices range over all positive semi-definite matrices satisfying certain	power constraints. This makes the complexity exponentially increase with $m$.	{As we have proved in the last section, the weight ordering in the WSR maximization problem determines the optimal encoding order.} Thus, Algorithm~\ref{alg:WSRalgorithm}  avoids an exhaustive search over all permutations which can save $K!$ complexity to reach the capacity region. It also finds the covariance matrices by an iterative algorithm, not an exhaustive search.

\textit{Comments on the upper bounds.} 
When $\mathbf{G}$ gets close to $\mathbf{0}$ or there is gradually no information leakage in the system, and the system approaches its upper bound.   The upper bound is the standard MIMO-BC \cite{weingartens2006capacity}. 

An important feature of the proposed algorithm is its generality and its potential for extension to other related problems, for example, the  Gaussian $K$-receiver interference channel \cite{ordentlich2014approximate} and covert communication over the $K$-receiver multiple access channel \cite{arumugam2019covert}.
The numerical algorithm can also work with more practical constraints and assumptions, such as per-antenna power constraints, where each antenna has its own power amplifier and is independently limited by the linearity of the amplifier \cite{yu2007transmitter}. In addition, this paper assumes perfect channel state information of receivers and eavesdropper, while a fast-fading channel with rapidly varying channel coefficients is more realistic, because of 	the limited feedback and delay caused by the 	channel estimation \cite{lin2014secrecy}. 

 \section{Two-receiver Case} \label{sec:tworeceiver}
The two-receiver wiretap channel is very typical and often is discussed independently  \cite{ekrem2010secure,  benammar2015secrecy, liu2010vector, bagherikaram2013secrecy}.  {We thus provide the details and also derive a closed-form solution of the iteration for the special case in this section.}

The capacity region is given in \eqref{scenario1}. The WSR formation can be obtained by setting $K=2$. The iterative covariance matrices in \eqref{reformuP3} can be written as 
\begin{subequations}\label{bsmm1}
	\begin{align} 
	\mathbf{Q}^{(i)}_1 &= {\rm arg} \max_{\mathbf{Q}_1 \succcurlyeq 0} f(	\mathbf{Q}_1, 	\mathbf{Q}^{(i-1)}_2),\label{bsmm1_1}\\
	\mathbf{Q}^{(i)}_2 &= {\rm arg} \max_{\mathbf{Q}_2 \succcurlyeq 0} f(	\mathbf{Q}^{(i)}_1, \mathbf{Q}_2). \label{bsmm1_2}
	\end{align}
\end{subequations}

\noindent The function $f(\mathbf{Q}_1, \mathbf{Q}_2)$ can be rewritten as the summation of  convex and  concave functions
\begin{align} \label{conv_concaa}
f(\mathbf{Q}_1, \mathbf{Q}_2) =f^{\rm ccv}_{k}(\mathbf{Q}_k, \mathbf{Q}_{\bar{k}}) + f^{\rm cvx}_{k}(\mathbf{Q}_k, \mathbf{Q}_{\bar{k}}),
\end{align}
in which
\begin{subequations}\label{conv_concab}
	\begin{align} 
	f^{\rm ccv}_{1}(\mathbf{Q}_1, \mathbf{Q}_{\bar{2}})  = & w_1 \log|\mathbf{I}_{n_1} +(\mathbf{I}_{n_1} + \mathbf{H}_1 \mathbf{Q}_2
	\mathbf{H}_1^\dag)^{-1}\mathbf{H}_1 \mathbf{Q}_1
	\mathbf{H}_1^\dag|\notag \\
	 -& \lambda{\rm tr}(\mathbf{Q}_1),  \label{ccv1}  \\
	f^{\rm cvx}_{1}(\mathbf{Q}_1,  \mathbf{Q}_{\bar{2}})  = & -w_1 \log|\mathbf{I}_{n_e} +(\mathbf{I}_{n_e} + \mathbf{G} \mathbf{Q}_2
	\mathbf{G}^\dag)^{-1}\mathbf{G} \mathbf{Q}_1
	\mathbf{G}^\dag|   \notag \\
	+& w_2\log \frac{|\mathbf{I}_{n_2} + \mathbf{H}_2 \mathbf{Q}_2 \mathbf{H}_2^\dag |}{|\mathbf{I}_{n_e} + \mathbf{G} \mathbf{Q}_2 \mathbf{G}^\dag |}  - \lambda({\rm tr}(\mathbf{Q}_2)-P). \label{cvx1}\\
	f_{2}^{\rm ccv}(\mathbf{Q}_2,  \mathbf{Q}_{\bar{1}})  =& w_1 \log|\mathbf{I}_{n_e} + \mathbf{G} \mathbf{Q}_2
	\mathbf{G}^\dag |  \notag \\ 
	+& w_2\log|\mathbf{I}_{n_2} + \mathbf{H}_2 \mathbf{Q}_2 \mathbf{H}_2^\dag |  - \lambda{\rm tr}(\mathbf{Q}_2), \label{ccv2} \\
	f^{\rm cvx}_{2}(\mathbf{Q}_2, \mathbf{Q}_{\bar{1}})  =& w_1 \log|\mathbf{I}_{n_1} +(\mathbf{I}_{n_1} + \mathbf{H}_1 \mathbf{Q}_2
	\mathbf{H}_1^\dag)^{-1}\mathbf{H}_1 \mathbf{Q}_1
	\mathbf{H}_1^\dag|  \notag \\
	 -& w_1\log|\mathbf{I}_{n_e} + \mathbf{G} (\mathbf{Q}_1 + \mathbf{Q}_2) \mathbf{G}^\dag |   \notag \\
	-& w_2\log|\mathbf{I}_{n_e} + \mathbf{G} \mathbf{Q}_2\mathbf{G}^\dag | - \lambda({\rm tr}(\mathbf{Q}_1)-P).  \label{cvx2}
	\end{align}
\end{subequations}

The partial derivatives of	$f^{\rm cvx}_{k}(\mathbf{Q}_k)$ with respect to $\mathbf{Q}_k$ are
\begin{align} 
\mathbf{A}_1  
=- w_1 \mathbf{G}^{\dag} (\mathbf{I}_{n_e} +\mathbf{G} (\mathbf{Q}_1 + \mathbf{Q}_2)
\mathbf{G}^{\dag})^{-1}\mathbf{G},  \label{eq:A}
\end{align} 
and
\begin{align} 
\mathbf{A}_2   &= + w_1 \mathbf{H}_1^{\dag} (\mathbf{I}_{n_1} +\mathbf{H}_1(\mathbf{Q}_1 + \mathbf{Q}_2)\mathbf{H}_1^{\dag})^{-1}\mathbf{H}_1 \notag \\
& -w_1 \mathbf{H}_1^{\dag} (\mathbf{I}_{n_1} +\mathbf{H}_1\mathbf{Q}_2\mathbf{H}_1^{\dag})^{-1}\mathbf{H}_1 \notag \\
& -w_1 \mathbf{G}^{\dag} (\mathbf{I}_{n_e} +\mathbf{G} (\mathbf{Q}_1 + \mathbf{Q}_2)
\mathbf{G}^{\dag})^{-1}\mathbf{G} \notag \\
& -w_2 \mathbf{G}^{\dag} (\mathbf{I}_{n_e} +\mathbf{G} \mathbf{Q}_2
\mathbf{G}^{\dag})^{-1}\mathbf{G}.  \label{eq:B}
\end{align} 
Then, we optimize the right-hand side of \eqref{reformuP21} by omitting the constant terms. Then problems \eqref{bsmm1_1} and \eqref{bsmm1_2} are re-expressed as 
\begin{align} \label{reformuP31}
	\mathbf{Q}^{(i)}_1= {\rm arg}\max \limits_{\mathbf{Q}_1} \; w_1 & \log|\mathbf{I}_{n_1} +(\mathbf{I}_{n_1} + \mathbf{H}_1 \mathbf{Q}_2
	\mathbf{H}_1^\dag)^{-1}\mathbf{H}_1 \mathbf{Q}_1
	\mathbf{H}_1^\dag| \notag \\
	& -  {\rm tr}[(\lambda \mathbf{I}_{n_t} - \mathbf{A}^{\dag}_{1})\mathbf{Q}_1], 
\end{align}
and
\begin{align} \label{reformuP41}
	\mathbf{Q}^{(i)}_2 &= {\rm arg }\max \limits_{\mathbf{Q}_2} \; w_1 \log|\mathbf{I}_{n_e} + \mathbf{G} \mathbf{Q}_2
	\mathbf{G}^\dag |   \notag \\ 
	& + w_2\log|\mathbf{I}_{n_2}+ \mathbf{H}_2 \mathbf{Q}_2 \mathbf{H}_2^\dag| -  {\rm tr}[(\lambda \mathbf{I}_{n_t} - \mathbf{A}^{\dag}_{2})\mathbf{Q}_2].
\end{align}
Both  \eqref{reformuP31}  and  \eqref{reformuP41} are convex problems and have optimal solutions. Different from the general formula in \eqref{reformuP3} which is solved using the rotation method, a closed-form solution for \eqref{reformuP31} can be obtained by the following lemma \cite{park2015weighted}. 
\begin{lem}{(\cite[Lemma 1]{park2015weighted})} \label{lemopt}
	For some $\mathbf{S} \succ \mathbf{0}$, the optimal solution of the function $\max_{\mathbf{Q}\succcurlyeq \mathbf{0}} w \log|\mathbf{I}+ \mathbf{R}^{-1}\mathbf{H}\mathbf{Q}\mathbf{H}^{\dagger}|-{\rm tr}(\mathbf{SQ})$ is given by 
	\begin{align}
	\mathbf{Q}^{*} = \mathbf{S}^{-1/2}\mathbf{V}\mathbf{\Lambda}\mathbf{V}^{\dagger}\mathbf{S}^{-1/2}.
	\end{align}
	To use Lemma \ref{lemopt}, we set  $w= w_1$, $\mathbf{S}=\lambda \mathbf{I}_{n_t} - \mathbf{A}^{\dag}$,  and $\mathbf{R}=\mathbf{I}_{n_1} +\mathbf{H}_1 \mathbf{Q}^{(i)}_2 \mathbf{H}_1^{\dag}$ for \eqref{reformuP31}. $\mathbf{V}$, $\mathbf{U}$, and $\mathbf{\Lambda}$ are obtained by eigenvalue decomposition of  $\mathbf{R}^{-1/2}\mathbf{H}\mathbf{S}^{-1/2}=\mathbf{U}{\rm diag}(\sigma_1, \sigma_2, \dots, \sigma_m)\mathbf{V}^{\dagger}$, $\sigma_i \geq 0$, $\forall i$, $\mathbf{\Lambda	}= {\rm diag}[(w -1/\sigma_1^2)^{+},  \dots, (w -1/\sigma_m^2)^{+}]$, and $(x)^{+}=\max(x,0)$.
\end{lem}
{The optimal covariance matrix of each iteration  in \eqref{reformuP31} can be obtained by the above lemma. 
 To achieve the secrecy rate region, we can  set $K=2$ in Algorithm~\ref{alg:WSRalgorithm}, and analytically solve for $\mathbf{Q}_1^{(i)}$  using  Lemma~\ref{lemopt}.}

\section{Simulation Results}\label{sec:simulation}

In this section, we provide three examples to illustrate the effect of the order of encoding and to compare  our signaling design with ZF beamforming. 
\begin{figure*}[t]
	\centering
	\begin{minipage}[t]{0.32\textwidth}
		\centering
		\includegraphics[width=6.2cm]{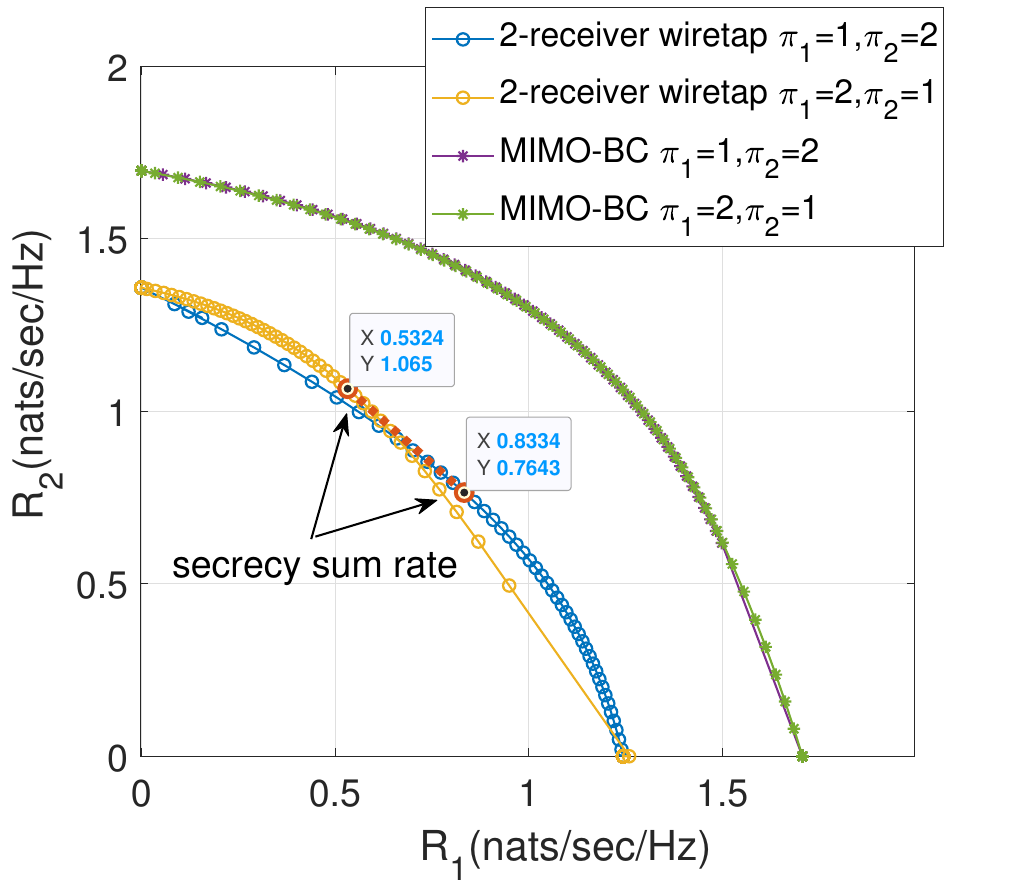}
		\caption{{Secrecy rate regions} for $K=2$. It is seen that  encoding is determining. }
		\label{fig:dpc1}
	\end{minipage} 
	\begin{minipage}[t]{0.32\textwidth}
		\centering
		\includegraphics[width=6.2cm]{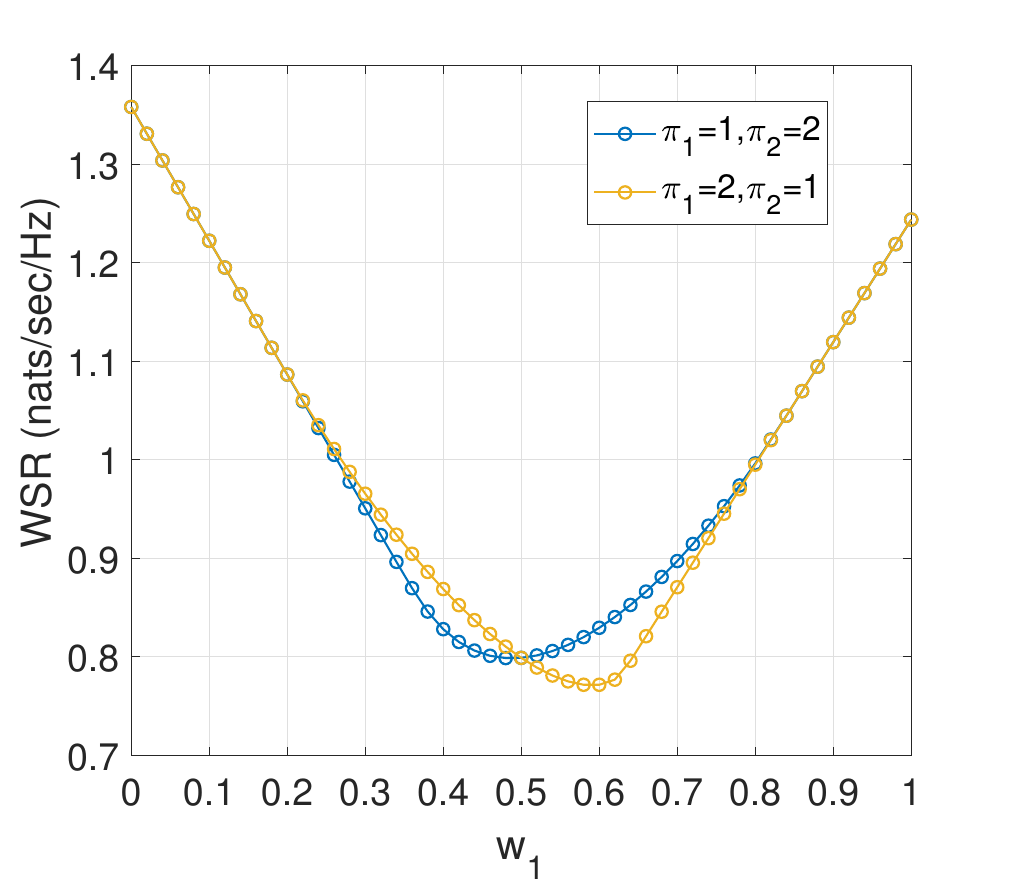}
		\caption{WSR versus the weight of the first user  for $K=2$. }
		\label{fig:order}
	\end{minipage} 
	\begin{minipage}[t]{0.32\textwidth}
		\centering
		\includegraphics[width=6.2cm]{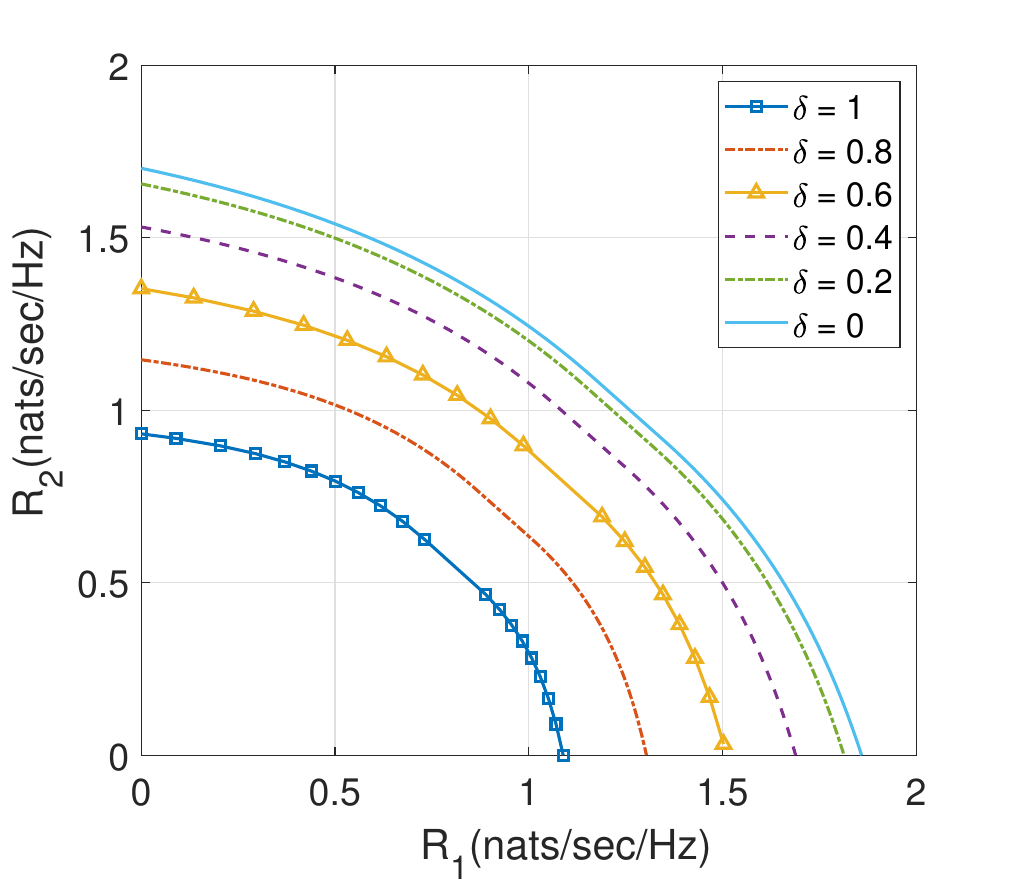}
		\caption{Secrecy rate region for different values of $\delta$  (the strength of Eve's channel). }
		\label{fig:dpc2}
	\end{minipage} 
\end{figure*} 

\textbf{Example 1:} We set the channels and power  the same as \cite{park2016secrecy}, which are given by
\begin{align}
\mathbf{H}_1& =\left[
\begin{matrix} 
1 &  -0.5 \\
0.5 &  2
\end{matrix}\right], \;
\mathbf{H}_2  =  \left[
\begin{matrix} 
-0.3 &  1 \\
2.0 &  -0.4
\end{matrix}\right], \notag \\
\mathbf{G}& =\left[
\begin{matrix}  0.8 & -1.6
\end{matrix}\right], \; P=1. \notag
\end{align}
In this case, $n_t=2$, {$n_1=n_2=2$,} and $n_e=1$. We set $w_1 = w_2=0.5$ to verify the subcase of the WSR, i.e.,  sum-rate. The accuracy parameters for Algorithm~\ref{alg:WSRalgorithm} are set as  $\epsilon_1 = \epsilon_2 = 10^{-5}$. Algorithm~\ref{alg:WSRalgorithm} achieves the  same results as \cite{park2016secrecy}. If receiver 1 is encoded first and then receiver 2,  i.e., $\pi_1 = 1$, $\pi_2 = 2$, 
the optimal secrecy rates are $(R_1, R_2) = (0.8334,  0.7643)$  $\rm{nats/sec/Hz}$ and the secrecy sum-rate is $1.5977$.  If receiver 2 is encoded first and then receiver 1, i.e., $\pi_1 = 2$, $\pi_2 = 1$, the optimal secrecy rates are $(R_1, R_2) = (0.5324,  1.065)$ which yields the same security sum-rate  $1.5977$.  It is seen that a different
	encoding order brings forth different optimal covariance matrices and rates, but the same secrecy sum-rate is achieved, as also noted in \cite{park2016secrecy}. 

If we vary the weight $w_1$ from 0 to 1 with a step 0.01, the  results of the rate pairs are shown in Fig.~\ref{fig:dpc1}. The curve marked with blue circles corresponds to the achievable rate region with ascending encoding order whereas the curve marked with yellow circles corresponds to the descending encoding order.  The purple and green curves represent the corresponding achievable rate regions of the MIMO-BC  \cite{weingartens2006capacity}, i.e., no eavesdropper or $\mathbf{G} = \mathbf{0}$. 

The ordering of the receivers' weights  in the WSR
maximization problem determines the optimal encoding order.  An example is shown in Fig.~\ref{fig:order}. If $w_1 > w_2$, the optimal encoding order is to encode 
message $M_1$ first and message $M_2$ next. If $w_1 < w_2$, the optimal encoding order is to encode $M_2$  first and  $M_1$ next.  If $w_1 = w_2$, the two orders will give the
same sum-rate but each order will give a different corner point
of the capacity region. See Fig.~\ref{fig:dpc1} and Fig.~\ref{fig:order} for illustrative
representations. 


Note that the boundary of the DPC achievable region has a straight-line segment for the secure sum-rate.  
This is because the secrecy capacity region is reached by the convex hull of all encoding orders of rate pairs \cite{ekrem2011secrecy}, \cite[Theorem 13]{ekrem2010secure}. This characteristic of possessing a straight-line segment, i.e.,
a time-division portion or time-sharing, also holds for the MIMO-BC \cite{vishwanath2003duality}. 
The slope for any point on the boundary is $-{w_2}/{w_1}$ because the WSR is $\varphi(P) =w_1 R_1+ w_2R_2$. 
Thus, the straight line with a slope $-1$ {passes} by the two sum-rate points.

\textbf{Example 2:} We consider another example where each node is equipped with two antennas. The channels are given by    
\begin{align} \label{example1}
\mathbf{H}_1 =\left[
\begin{matrix} 
1 &  0.8 \\
0.5 &  2
\end{matrix}\right], \;
\mathbf{H}_2  =  \left[
\begin{matrix} 
0.2 &  1 \\
2 &  0.5
\end{matrix}\right], \;
\mathbf{G} =\delta \left[
\begin{matrix}  1 & 0.4\\
-0.4 &  1
\end{matrix}\right]. 
\end{align}
$\mathbf{H}_1$ and $\mathbf{H}_2$ are set the same as that in \cite{vishwanath2003duality}, and $P=1$. The parameter $\delta$ indicates the strength of Eve's channel. 
The achievable (secure) rate regions are shown in Fig.~\ref{fig:dpc2}. 
When $\delta=0$, the eavesdropping channel becomes $\mathbf{G} = \mathbf{0}$, and thus the problem becomes the MIMO-BC and forms an upper bound.  As $\delta$ increases,  the channel of the eavesdropper gets stronger,  and the secure rate region begins to shrink. 
Note that the rate region for the MIMO-BC ($\delta=0$) can be realized by applying the proposed method and Algorithm~\ref{alg:WSRalgorithm} by setting $\mathbf{G} = \mathbf{0}$. This leads to the same results provided by MAC-BC duality \cite{vishwanath2003duality}. The sum-rate for the MIMO-BC in this case is also  $2.2615$, which is the same as that in \cite[Example 1]{vishwanath2003duality}.

In addition, we show how much power is allocated to each receiver for different values of $w_1$ and $w_2=1-w_1$ in the case with $\delta=1$. When $w_1$ becomes larger, user~1 is allocated more power while user~2 is allocated less power. Interestingly, when $w_1<0.11$ and the optimal encoding order is user~2 first then user~1 (corresponding to the red points in Fig.~\ref{fig:dpc5}), user~2 is allocated all of the power. When $w_1>0.86$ and the optimal encoding order is user~1 first then user~2 (corresponding to the blue points in Fig.~\ref{fig:dpc5}), all power is allocated to receiver~1. 
 
\begin{figure}[t]
	\centering
	\includegraphics[width=0.45\textwidth]{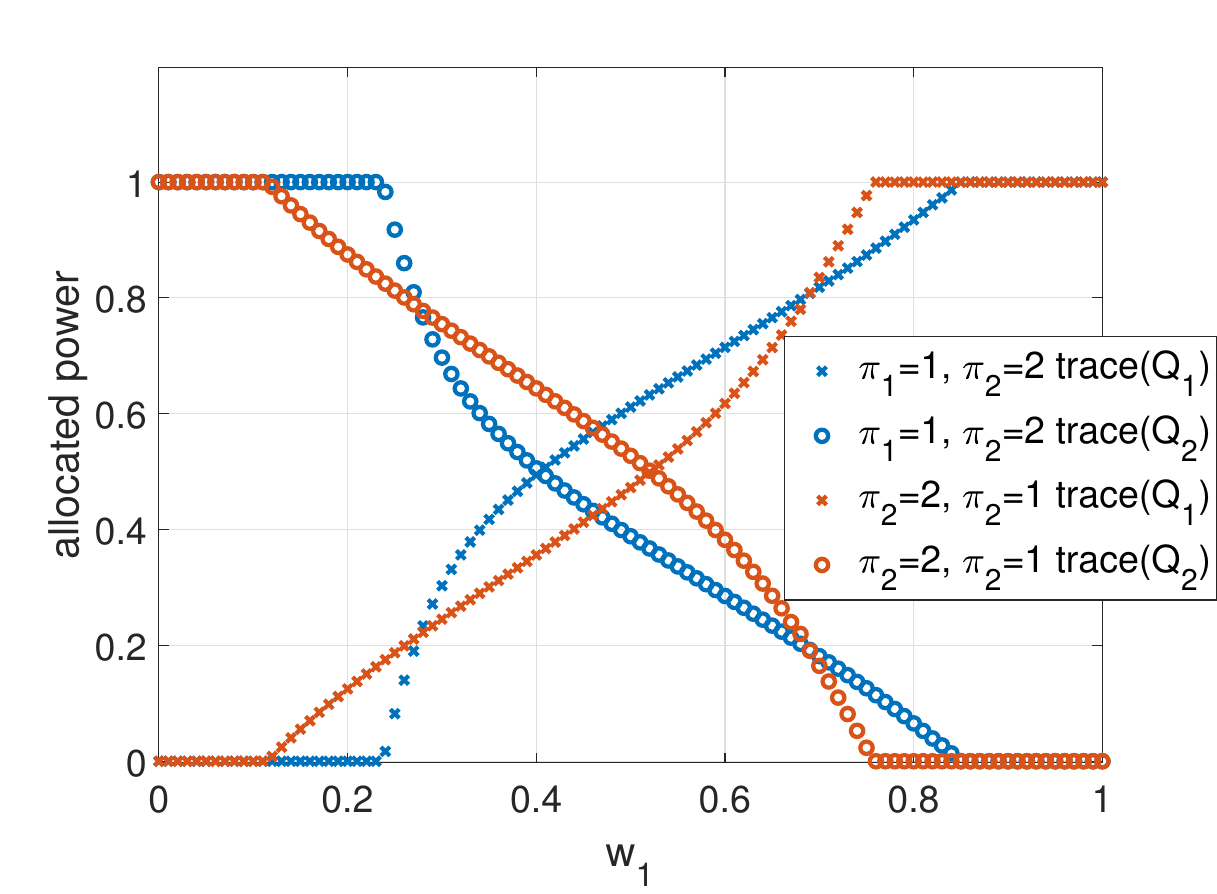}
	\caption{{Power allocated to each user of the two-user wiretap channel for different orders of encoding}.}
	\label{fig:dpc5}
\end{figure}

\textbf{Example 3:} We consider $K=3$, in which the channels are 
\begin{align} \label{example2}
\mathbf{H}_1& =\left[
\begin{matrix} 
-0.4332 + 0.7954i &  -0.3152 - 1.8835i \\
-1.0443 + 1.2282i &  -0.2614 + 0.2198i
\end{matrix}\right], \\
\mathbf{H}_2&  =  \left[
\begin{matrix} 
   1.3389 - 0.5995i &  -0.6924 - 0.4542i \\
-1.2542 + 0.1338i &  -2.1644 + 0.6520i
\end{matrix}\right], \notag \\
\mathbf{H}_3&  =  \left[
\begin{matrix} 
 1.0291 - 0.0212i &    -0.3016 - 0.3662i \\
 0.1646 + 0.5179i &    0.3075 + 0.2919i
\end{matrix}\right], \notag \\
\mathbf{G}& = \left[
\begin{matrix} 
  -0.0875 - 0.9443i &   -0.4637 + 0.7799i
\end{matrix}\right]. 
\end{align}

\begin{figure}[h]
	\centering
	\includegraphics[width=0.4\textwidth]{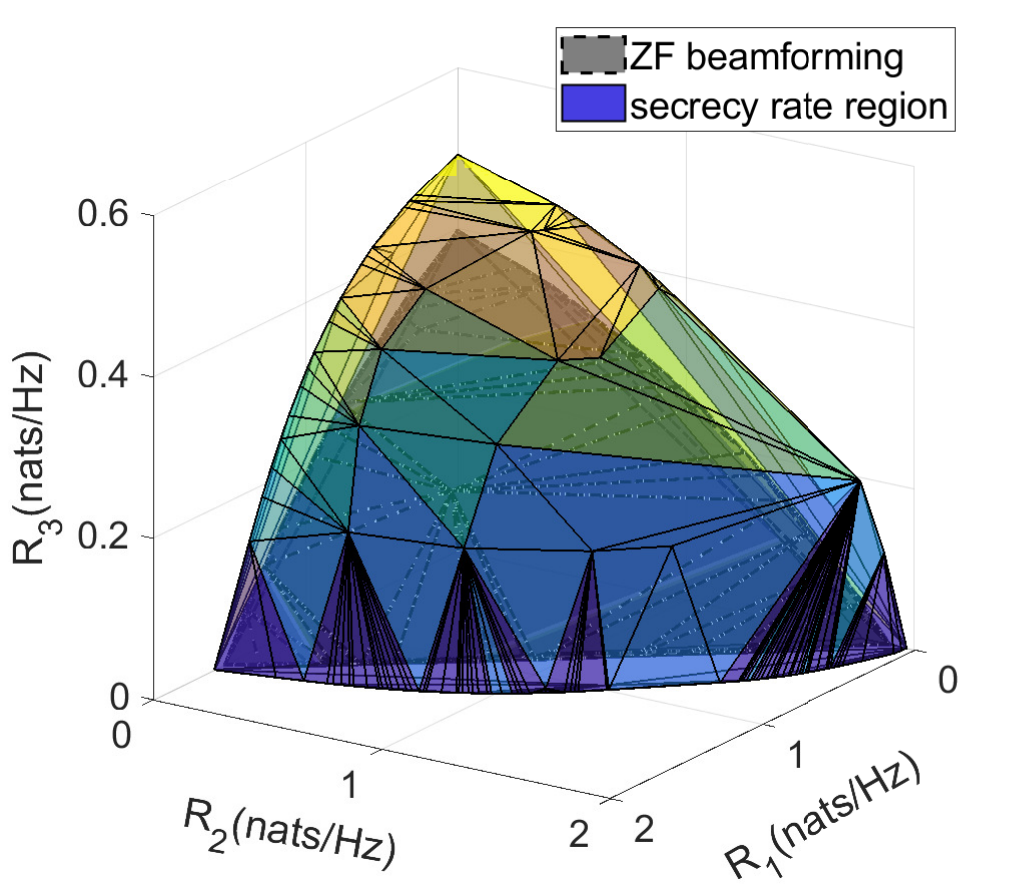}
	\caption{{Secrecy rate region for a three-receiver wiretap channel}. }
	\label{fig:3receivers}
\end{figure}
In this case, $P=1$, $n_t=2$, $n_1=n_2=n_3=2$, and $n_e=1$.
Figure~\ref{fig:3receivers} shows the results of the secrecy achievable rate regions achieved by Algorithm~\ref{alg:WSRalgorithm} and ZF beamforming. 
To guarantee confidentiality, ZF beamforming  \cite{park2016secrecy, tang2016low} projects
each channel $\mathbf{H}_k$ in \eqref{example2} into the null space of $\mathbf{G}$ such that $\mathbf{G}\mathbf{Q}_k = \mathbf{0}$. Hence, the
achievable secrecy rate region is reduced to the WSR maximization of MIMO BC. The projected channels are obtained by the inner product of the channels  \eqref{example2} with null space of  $\mathbf{G}$, i.e.,  
\begin{align*} 
\mathbf{h}_1& =\left[
\begin{matrix} 
 -0.4459 - 1.2743i &  -1.3394 - 0.0374i
\end{matrix}\right]^{\dagger}, \\
\mathbf{h}_2&  =  \left[
\begin{matrix} 
 0.5055 - 0.0275i    &    -2.2319 + 0.6324i 
\end{matrix}\right]^{\dagger}, \notag \\
\mathbf{h}_3&  =  \left[
\begin{matrix} 
   0.3099 - 0.2412i  &      -0.0639 - 0.5733i
\end{matrix}\right]^{\dagger}.
\end{align*}
The proposed method achieves a larger rate  region compared with ZF beamforming, which is  known to be a {suboptimal scheme \cite{park2016secrecy}}. 
If we set any  of the receivers to idle, the problem reduces to two-receiver case. The projection  of the secrecy capacity onto  $(R_1, R_2)$, $(R_2, R_3)$, and $(R_1, R_3)$ planes gives {the capacity region of the corresponding two-receiver wiretap channel.}

\begin{figure*}[t]
	\centering
	\includegraphics[width=0.9\textwidth]{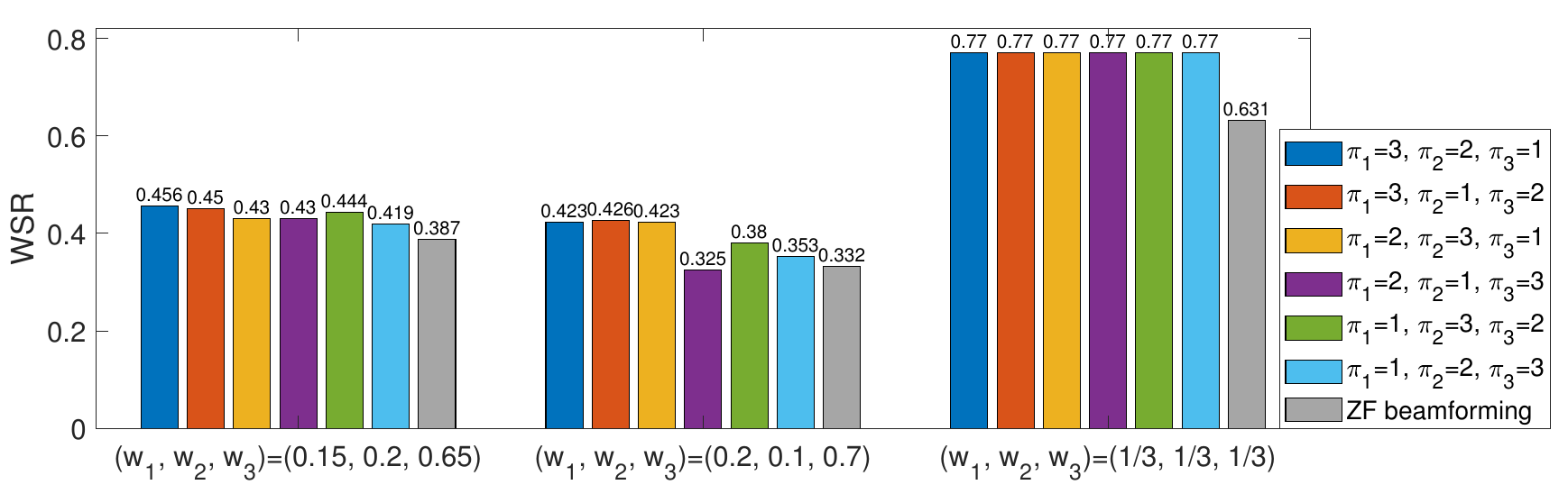}
	\caption{The WSR for the three-receiver wiretap channel {with all six permutations of encoding order as well as ZF for three sets of weights.}} 
	\label{fig:dpc4}
\end{figure*}

Figure~\ref{fig:dpc4} shows {the WSR for three weight allocations and different coding orders.}
Each bar reflects the secure WSR of the system under a specific weighting coefficient. The maximum WSR is achieved with the encoding order $\pi_1 \ge \pi_2 \ge  \dots \ge \pi_K$. For example, for $(w_1, w_2, w_3)=(0.15, 0.2, 0.65)$, the highest WSR is obtained for $\pi_1 = 3$, $\pi_2 = 2$, $\pi_3 = 1$ whereas for  $(w_1, w_2, w_3)=(0.2, 0.1, 0.7)$ the highest WSR is obtained for $\pi_1 = 3$, $\pi_2 = 1$, $\pi_3 = 2$.   {The sum-rate achieved by Algorithm~\ref{alg:WSRalgorithm} is $0.77$ ${\rm nats/sec/Hz}$ for all permutations with $(w_1, w_2, w_3)=(1/3, 1/3, 1/3)$, which verifies Corollary~\ref{corollary}.} 
{ The WSR acived by ZF is far from the proposed signaling design with optimal encoding order. }

\section{Conclusions}\label{sec:conclusion}

The optimal order for the $K$-receiver MIMO-BC with an eavesdropper has been established in this paper. The proof has been achieved by transforming the original WSR maximization problem of the MIMO-BC wiretap channel into an equivalent problem using BC-MAC duality and solving the new problem.  The proof reveals that a receiver having a higher weight in the WSR problem should be encoded earlier. This finding reduces the complexity of determining the capacity region  $K!$ times. In addition, we have solved the WSR maximization problem to find the optimal signaling  for this channel. The secrecy rate equations are nonconvex functions of transmit covariance matrices. However, we have proven that the duality gap is  zero, and KKT conditions are necessary for the solutions. We then use the BSMM to find the covariance matrices in an iterative manner.  In the two-receiver case, we find a closed-form solution for each iteration. Numerical results verify the optimal encoding order. The proposed method outperforms existing solutions, like ZF beamforming, and is more efficient than applying an exhaustive search to the DPC region.

\appendices
\section{Proof of Lemma~\ref{lemma1}} \label{proofLemma}
Based on Theorem 1 in \cite{yu2006dual}, if  the optimization problem in \eqref{eq_WSR}  satisfies the \textit{time-sharing property},  then it has a zero duality gap. Thus, our goal is to prove that the optimization problem \eqref{eq_WSR} satisfies the time-sharing property. If the optimal $\varphi(P)$ in \eqref{eq_WSR}  is a  concave
function of the total power $P$, then the time-sharing property holds \cite{yu2006dual}. 
{
Let $P_1, P_2, \dots, P_K \geq 0$ be the power constraints and {there exists $\hat{P}=\sum_{k=1}^{K}\theta_k P_k$} for some $0 \leq\theta_k \leq 1$ and $\sum_{k=1}^{K} \theta_k =1$.
By contradiction, assume that the time-sharing does not hold for fixed non-negative weights. Then,  {$\varphi(\hat{P})$} must be a convex function, i.e., 
\begin{align} \label{timesharing}
\varphi(\hat{P}) = \varphi(\sum_{k=1}^{K}\theta_k P_k) \le  \sum_{k=1}^{K} \theta_k \varphi(P_k).
\end{align}
Next, assume that the rate tuple $(R^{(k)}_1, R^{(k)}_2, \dots, R^{(k)}_K)$ is achieved by solving \eqref{eq_WSR} for $P=P_k$. Further, let us allocate a fraction $\theta_k$ of time to user $k$.  Time-division multiplexing gives the rate tuple  $$\left(\sum_{k=1}^{K}\theta_k R^{(1)}_k, \dots,  \sum_{k=1}^{K}\theta_k R^{(k)}_k, \dots, \sum_{k=1}^{K}\theta_k R^{(K)}_k \right),$$ and the WSR is the  right-hand side of \eqref{timesharing}.
This implies that the time division scheme with an average power {over a unit time slot  $\bar{P} = \sum_{k=1}^{K}\theta_k P_k$}  achieves a point higher than the secrecy capacity $\mathcal{C}(\{\mathbf{H}_k\}_{k=1}^{K}, \mathbf{G}, \bar{P})$. This is a contradiction because the convex closure operator cannot increase the secrecy
capacity region. Thus, {$\bar{P}=\hat{P}$} and the optimization problem \eqref{eq_WSR} is a concave function of $P$ which implies the time-sharing property.} 
Thus, zero duality gap is guaranteed for the optimization problem \eqref{eq_WSR}
and the KKT conditions are necessary for the optimal solution \cite[Chapter 5, pp. 243]{boyd2004convex}. The capacity is then achieved by the convex closure of the union of the aforementioned rate regions over all possible permutations  $\pi$. This convex hull operation  performs a time-division  multiplexing and also satisfy the time-sharing property.

\bibliography{proposal20220304}
\bibliographystyle{IEEEtran}

\begin{IEEEbiography}
	[{\includegraphics[width=1in,height=1.25in,clip,keepaspectratio]{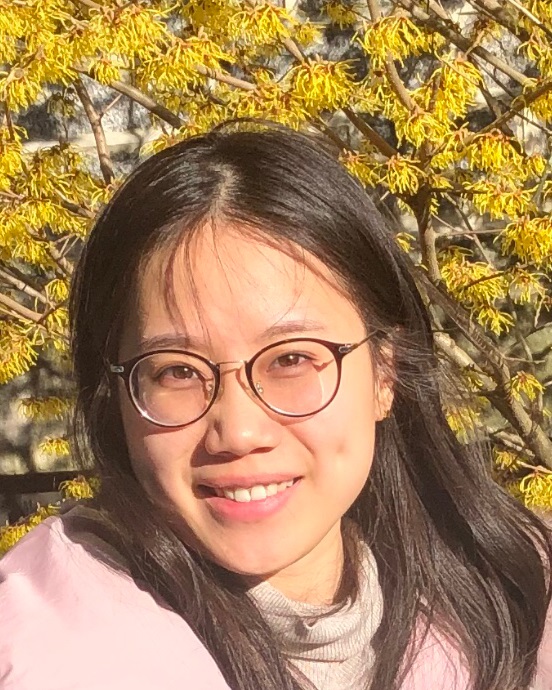}}]{Yue Qi}
	(S'20–M’22) received the B.Eng. and the M.Eng. degrees in Electronic Engineering from Xidian University, and the Ph.D. degree in Electrical and Computer Engineering from Villanova University. Her research interests include NOMA, physical layer security, and signal processing. She is a recipient of the IEEE Communication Society Student Grant in ICC'20, 22. She is currently working as 
	a Senior Research Engineer with the Standards
	and Mobility Innovation Laboratory, Samsung
	Research America, Plano, TX, USA.  
\end{IEEEbiography}

\begin{IEEEbiography}
	[{\includegraphics[width=1in,height=1.25in,clip,keepaspectratio]{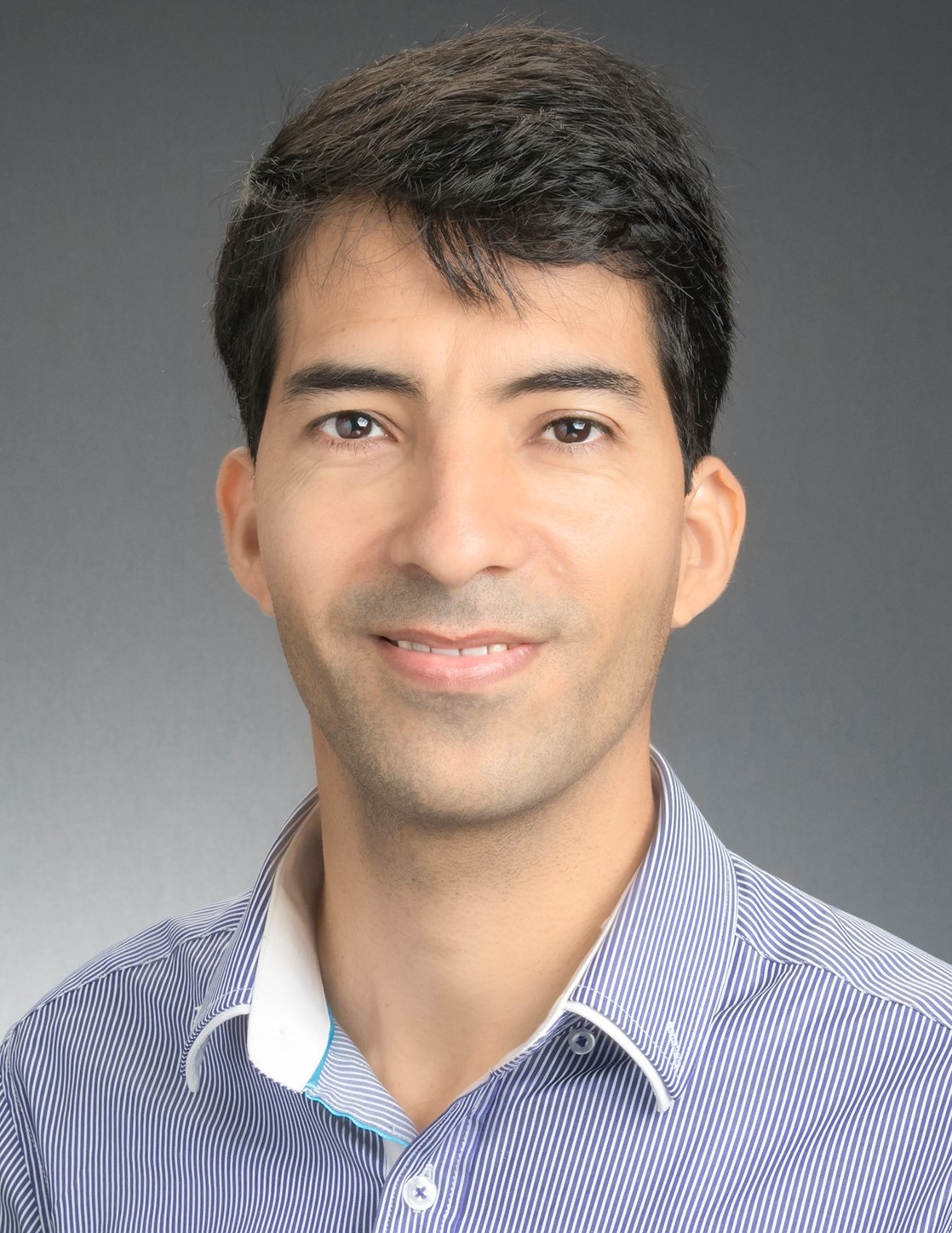}}]{Mojtaba Vaezi} (S’09–M’14–SM’18) received the B.Sc. and M.Sc. degrees from Amirkabir University of Technology (Tehran Polytechnic) and the Ph.D. degree from McGill University, all in Electrical Engineering. From 2015 to 2018, he was with Princeton University as a Postdoctoral Research Fellow and Associate Research Scholar. He is currently an Assistant Professor of ECE at Villanova University. Before joining Princeton, he was a researcher at Ericsson Research in Montreal, Canada. His research interests include the broad areas of signal processing and machine learning for wireless communications with an emphasis on physical layer security and fifth-generation (5G) and beyond radio access technologies. Among his publications in these areas is the book \textit{Multiple Access Techniques for 5G Wireless Networks and Beyond} (Springer, 2019). 
	
	Dr. Vaezi is an Editor of \textsc{IEEE Transactions on Communications} and \textsc{IEEE Communications Letters}. He has co-organized six NOMA workshops at IEEE VTC 2017-Spring, Globecom’17, 18, and ICC’18, 19, 20. He is a recipient of several academic, leadership, and research awards, including McGill Engineering Doctoral Award, IEEE Larry K. Wilson Regional Student Activities Award in 2013, the Natural Sciences and Engineering Research Council of Canada (NSERC) Postdoctoral Fellowship in 2014, Ministry of Science and ICT of Korea’s best paper award in 2017, IEEE Communications Letters Exemplary Editor Award in 2018,  the 2020 IEEE Communications Society Fred W. Ellersick Prize, and, the 2021 IEEE Philadelphia
	Section Delaware Valley Engineer of the Year Award. 
\end{IEEEbiography}	

\begin{IEEEbiography}
	[{\includegraphics[width=1in,height=1.25in,clip,keepaspectratio]{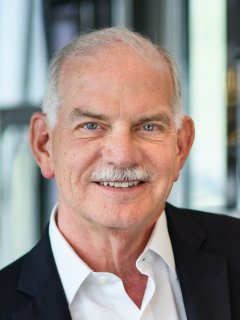}}]{H. Vincent Poor}  (S’72, M’77, SM’82, F’87) received the Ph.D. degree in EECS from Princeton University in 1977.  From 1977 until 1990, he was on the faculty of the University of Illinois at Urbana-Champaign. Since 1990 he has been on the faculty at Princeton, where he is currently the Michael Henry Strater University Professor. During 2006 to 2016, he served as the dean of Princeton’s School of Engineering and Applied Science. He has also held visiting appointments at several other universities, including most recently at Berkeley and Cambridge. His research interests are in the areas of information theory, machine learning and network science, and their applications in wireless networks, energy systems and related fields. Among his publications in these areas is the recent book Machine Learning and Wireless Communications.  (Cambridge University Press, 2022). Dr. Poor is a member of the National Academy of Engineering and the National Academy of Sciences and is a foreign member of the Chinese Academy of Sciences, the Royal Society, and other national and international academies. He received the IEEE Alexander Graham Bell Medal in 2017.
\end{IEEEbiography}

\end{document}